\documentclass[a4paper,11pt]{article}
\usepackage[utf8]{inputenc}
\usepackage{lmodern}
\usepackage[T1]{fontenc}
\usepackage{amsmath,amsthm,amsfonts,amssymb,bm,mathtools}
\mathtoolsset{showonlyrefs=true}

\usepackage{dsfont}			
\usepackage{enumitem}
\usepackage{graphicx,psfrag,epsf}
\usepackage{upgreek}
\graphicspath{ {./figures/} }
\usepackage{booktabs}
\usepackage{subcaption}
\usepackage{xcolor}
\usepackage[authoryear]{natbib}
\usepackage[right=2cm,left=2cm,top=2cm,bottom=3cm]{geometry}
\usepackage[font={small,it}]{caption}	
\usepackage[colorlinks,linkcolor=blue,citecolor=blue,urlcolor=blue]{hyperref}

\usepackage{multirow}
\usepackage{makecell}
\usepackage{afterpage}

\def\bN{{\mathbb N}}
\def\bfY{\mathbf{Y}}

\usepackage{amsthm}
\newtheorem{theorem}{Theorem}[section]

\newtheorem{proposition}[theorem]{Proposition}
\newtheorem{remark}[theorem]{Remark}

\usepackage{xcolor}
\usepackage{hyperref}
\usepackage{makecell}
\usepackage{longtable}
\usepackage{tabularx}
\usepackage{subcaption}
\usepackage{booktabs}
\def\bN{{\mathbb N}}
\def\bfY{\mathbf{Y}}
\usepackage{bbm,amsmath,amsfonts,amssymb,bm,mathtools}
\usepackage{amsthm}

\def\tsc#1{\csdef{#1}{\textsc{\lowercase{#1}}\xspace}}
\tsc{WGM}
\tsc{QE}
\title{\bf Data compression for fast dimension reduction and clustering of high-dimensional discrete data}
 \author{Silvia D'Angelo\\
  School of Computer Science and Statistics, Trinity College Dublin\\
   and \\
   Michael Fop \\
  School of Mathematics and Statistics, University College Dublin}
\date{}

\begin{document}
\maketitle

\begin{abstract}
High-dimensional discrete data are common in genomics, microbiomics, survey research, and digital behavioral analysis. Clustering such data is challenging because many existing methods are computationally expensive, sensitive to sparsity and discreteness, or designed for specific data types. We introduce a deterministic dimension-reduction framework for clustering high-dimensional discrete observations. The approach compresses observations into a low-dimensional continuous representation using weighted sums derived from a scaled positional encoding, yielding a numerically stable transformation applicable to both binary and count data. Several theoretical properties are established. The compression mapping is injective, ensuring that distinct observations remain distinguishable after transformation. Under mild regularity conditions, the compressed variables are approximately Gaussian, supporting the use of model-based clustering in the reduced space. We further show that separation between cluster centroids is preserved, indicating that location-based cluster structure remains identifiable following dimension reduction. Simulation studies demonstrate accurate cluster recovery across diverse settings, while achieving substantial computational savings compared with commonly used dimension-reduction techniques. Applications to microbiome data and United Nations roll-call voting data highlight the method's practical utility. Overall, the framework offers a scalable, efficient, and broadly applicable solution for clustering high-dimensional discrete data. 
\end{abstract}




\section{Introduction}
\label{sec:intro}

Clustering is a fundamental tool for identifying groups of observations that exhibit similar patterns across measured variables. It is widely used in exploratory analysis, classification, dimension reduction, and scientific discovery \citep{gormley:2023,wade:2023,dinh:2025}. Increasingly, clustering problems involve high-dimensional discrete data, including categorical, binary, and count-valued observations where the number of variables may be comparable to or exceed the sample size. Examples arise in single-cell RNA sequencing \citep{kiselev:2019,grabski:2023}, microbiome studies \citep{Shi:2022,chen:2025}, latent class analysis of survey data \citep{weller:2020}, and social media analytics \citep{champon:2026}. See \cite{dinh:2025} for a review of clustering methods for categorical data.

High dimensionality creates substantial challenges for clustering. As the number of variables increases, distances become less informative, noise variables can overwhelm signal, and flexible models become overparameterized. Numerous methods have been proposed for continuous data \citep{bouveyron:2014}, including sparse Bayesian Gaussian mixtures \citep{yao:2025}, latent factor and repulsive mixture models \citep{ghilotti:2025,chandra:2023}, scalable Gaussian mixtures with outliers \citep{zhou:2025}, random-projection approaches \citep{mori:2026,anderlucci:2022}, nonparametric clustering \citep{wang:nonpar:2025}, and spectral graph clustering \citep{yang:2021}. However, these methods are primarily designed for continuous observations.

For discrete data, existing approaches include partitional, distance-based, spectral, and model-based methods. The k-modes algorithm extends k-means to categorical observations \citep{huang:1998}, while distance-based methods rely on similarity measures tailored to discrete values \citep{guha:2000}. Spectral methods have been adapted to binary and categorical data, often with theoretical recovery guarantees \citep{tian:2024,lyu:2026}. Model-based approaches include latent class analysis \citep{weller:2020}, Dirichlet-multinomial mixtures \citep{bouguila:2009}, binary-data mixtures \citep{papastamoulis:2017,tang:2015,failli:2025}, Bayesian clustering based on Hamming distances \citep{argiento:2025}, variational Bayes mixtures with variable selection \citep{rao:2025}, multinomial count-data mixtures \citep{papastamoulis:2023}, copula-based clustering \citep{brini:2025}, feature-screening approaches for single-cell RNA sequencing \citep{Wang:2025}, and Poisson-lognormal or logistic-normal multinomial models \citep{payne:2025,fang:2023}.

Despite their diversity, these methods face common challenges in genuinely high-dimensional settings. Partitional algorithms can be sensitive to initialization, distance-based methods may struggle with ties and sparsity, spectral methods require computationally intensive matrix decompositions, and model-based approaches often rely on restrictive distributional assumptions and become costly as dimension increases. Consequently, performance frequently deteriorates when the number of variables $p$ greatly exceeds the sample size $n$.

A complementary strategy is to reduce the dimension of the data before clustering. 
For continuous data, principal component analysis, random projections, and spectral embeddings are routinely used as pre-clustering steps, and in practice they are often applied to discrete data as well. In high-dimensional settings, however, these methods can fail to produce a representation that separates the clusters: when the cluster signal is weak relative to the noise accumulated over many variables, the leading directions of variation and the pairwise dissimilarities are dominated by noise rather than by the cluster structure \citep[][see also the example below]{chang:1983,hall:2005,johnstone:2009,fan:2014}.

For discrete data, dimension reduction is further complicated because numerical embeddings can distort structure, while distance-based approaches may be sensitive to sparsity, ties, and the choice of dissimilarity. Existing proposals include ensemble dissimilarity matrices \citep{amiri:2018}, spectral projections \citep{tian:2024}, rank-based methods \citep{lausser:2018}, and multidimensional scaling for count data \citep{chen:2025}. Most, however, require an $n \times n$ dissimilarity matrix, involve iterative or stochastic optimization, or are tailored to specific data types.

An alternative approach is to compress the data. In the clustering literature, compression has been used to define similarity or a clustering criterion \citep{cilibrasi:2005, ma:2007}, and to reduce the sample by replacing the $n$ observations with summaries, coresets, or sketches \citep{feldman:2020, keriven:2018}. Closest to our setting are methods that compress the variables through a map fixed in advance, which avoids both pairwise dissimilarities and iterative optimization. Random projections \citep{johnson:1984} approximately preserve Euclidean distances and have been combined with $k$-means and model-based clustering \citep{fern:2003, boutsidis:2015, anderlucci:2022}. Hashing and sketching schemes for sparse discrete data, such as feature hashing \citep{weinberger:2009}, count sketches \citep{charikar:2004, clarkson:2017}, and sketches for binary and categorical vectors \citep{pratap:2019, verma:2022, bera:2023}, assign each variable to one of a small number of buckets and aggregate within buckets. These methods are randomized, hence the representation changes from run to run, and the discrete sketches require dedicated estimators to recover distances. Moreover, because they are designed to preserve the original distances or inner products, they inherit the limitation noted above: when the cluster signal is weak relative to the accumulated noise, the preserved distances may themselves be uninformative. A simple, deterministic, and broadly applicable compression approach that returns continuous coordinates directly usable by standard clustering methods remains lacking.

We address this gap by proposing a deterministic compression framework for clustering high-dimensional discrete data. Each observation is mapped to a low-dimensional continuous representation through weighted sums of the original variables using a numerically stable scaling of a base-$K$ positional encoding. The method is fast, deterministic, and applicable to binary, categorical, and count-valued data. Unlike model-based approaches, it does not require fitting a high-dimensional mixture model; unlike most dimension-reduction methods, it avoids pairwise dissimilarities, eigenvalue decompositions, and stochastic optimization; and unlike random projections and hashing-based sketches, it involves no randomness and needs no dedicated distance estimator. The compressed data representation can be used directly as input to standard clustering procedures, including k-means, Gaussian mixture models, or distance-based methods, making the proposal useful both as a standalone clustering pipeline and as a fast initialization tool for more complex clustering models.

The proposed compression has several useful properties. The mapping is injective, preventing distinct observations from collapsing to identical compressed values. Under mild moment and independence conditions, compressed coordinates are approximately Gaussian, providing a rationale for Gaussian mixture modeling in the compressed space. Furthermore, separation between cluster centroids is preserved, ensuring that location-based cluster structure remains identifiable after compression. This separation grows with the number of variables combined in each compressed coordinate, so that small but consistent differences spread over many variables accumulate.

As an illustrative example, consider $n = 50$ observations on $p = 10{,}000$ count variables, drawn from two groups with equal probability. Conditionally on group membership, the entries are independent, with $y_{ij} \sim \mathrm{Poi}(0.30)$ in the first group and $y_{ij} \sim \mathrm{Poi}(0.35)$ in the second. The resulting data matrix is sparse: about 72\% of the entries are zero and counts range from $0$ to $7$. The groups are highly overlapping, as the difference in means is 0.05 (less than a tenth of the within-group standard deviation) and individual variables carry almost no information about group membership (Figure~\ref{fig:intro-example}a). The two-dimensional representations obtained by principal component analysis (PCA) and multidimensional scaling (MDS) show no visible separation between the groups (Figure~\ref{fig:intro-example}b--c), as the weak signal in each variable is swamped by the noise accumulated over thousands of sparse variables. In contrast, the proposed compression to $q = 2$ dimensions separates the two groups clearly (Figure~\ref{fig:intro-example}d), because the small differences in means accumulate within each compressed coordinate. The compression is also faster: in R, PCA (\texttt{prcomp}) and MDS (\texttt{cmdscale}) take about 12 and 6 times longer, respectively.

\begin{figure*}[t!]
    \centering
    \begin{subfigure}[t]{0.5\textwidth}
        \centering
        \includegraphics[scale=0.5]{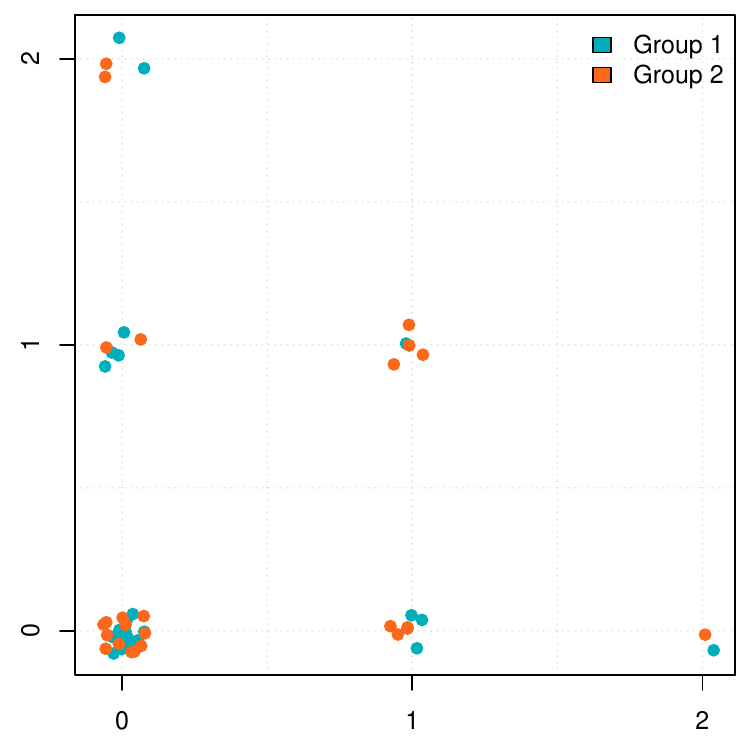}
        \caption{Original data}
    \end{subfigure}%
    ~ 
    \begin{subfigure}[t]{0.5\textwidth}
        \centering
        \includegraphics[scale=0.5]{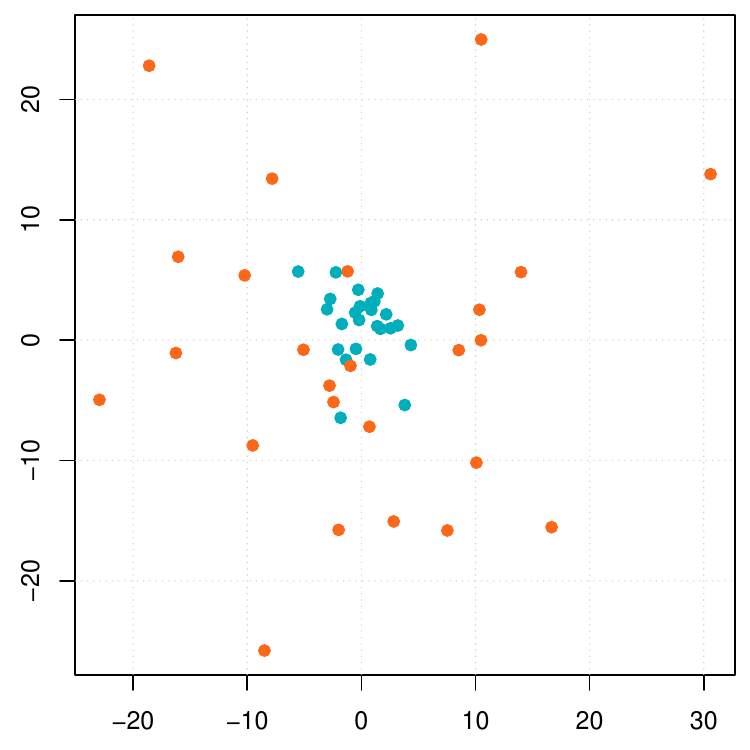}
        \caption{PCA}
    \end{subfigure}
    \\\medskip
    \begin{subfigure}[t]{0.5\textwidth}
        \centering
        \includegraphics[scale=0.5]{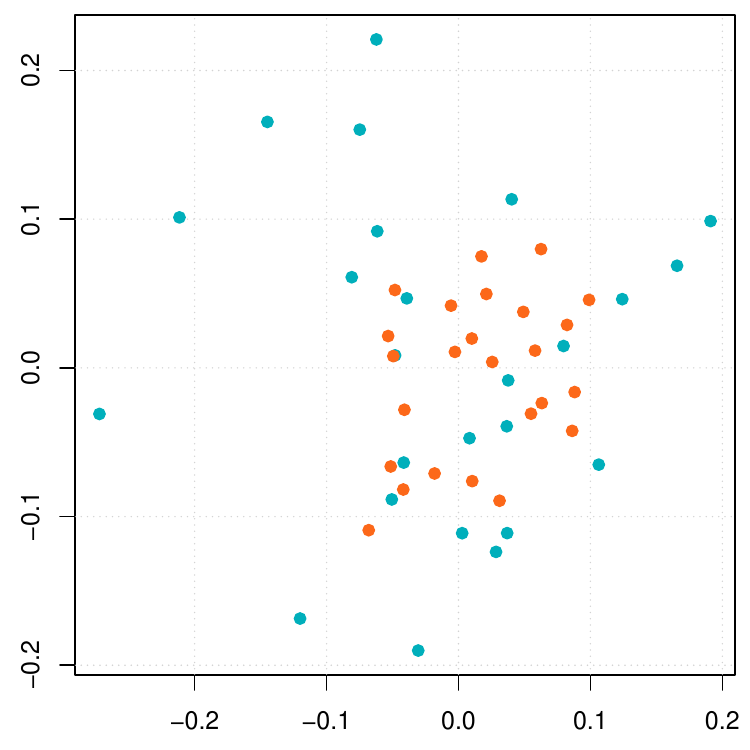}
        \caption{MDS}
    \end{subfigure}%
    ~ 
    \begin{subfigure}[t]{0.5\textwidth}
        \centering
        \includegraphics[scale=0.5]{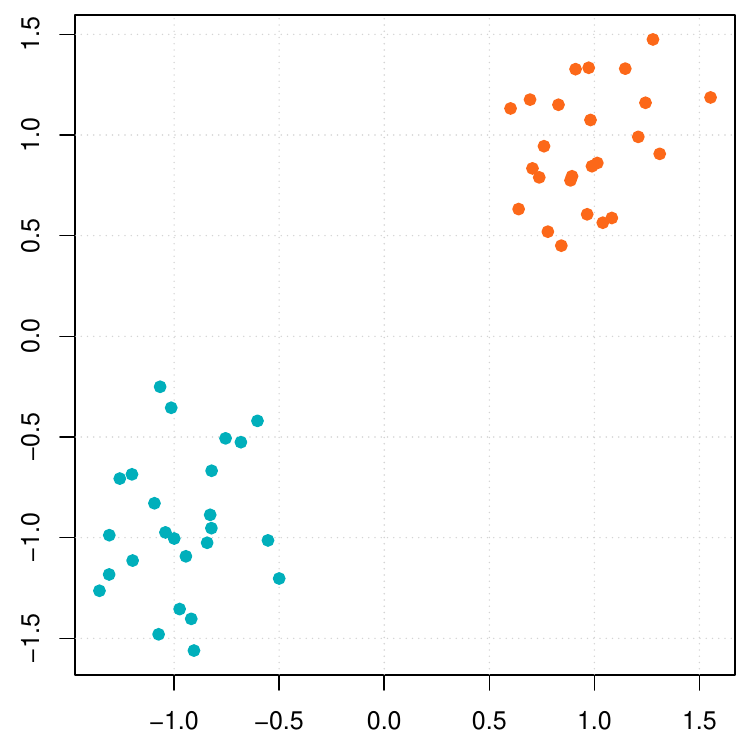}
        \caption{Data compression}
    \end{subfigure}
    
    \caption{\label{fig:intro-example} Illustrative example with sparse count data with overlapping classes. (a) Two of the $p$ variables jittered to show overlapping points. (b)--(d) Two-dimensional representations of the data obtained by: (b) PCA, (c) classical MDS based on the Bray-Curtis dissimilarity, and (d) the proposed data compression with $q = 2$. Points are colored by group membership.}
\end{figure*}

The remainder of the paper is organized as follows. Section~\ref{sec:meth} introduces the compression method and its theoretical properties. Section~\ref{sec:clust_compression_space} discusses clustering in the compressed space. Section~\ref{sec:sim} presents simulation results, and Section~\ref{sec:application} illustrates the method using United Nations voting data and a microbiome dataset comparing Hadza hunter-gatherers with Italian urban adults. Section~\ref{sec:discussion} concludes.

\section{Data compression and clustering}
\label{sec:meth}

Let $\mathbf{Y}\in\mathbb{N}^{n\times p}$ be a data matrix, where each column records the realizations of a discrete random variable taking values in $\{0,\dots,K-1\}$, for some $K\in\bN$, with $K\geq 2$. Let $\mathbf{y}_i=(y_{i1},\ldots,y_{ij}, \ldots, y_{ip})$ denote the $i^{\text{th}}$ row of $\mathbf{Y}$. Because each entry lies in the discrete interval $[0, K-1]$, the row $\mathbf{y}_i$ may be interpreted as a $p$-digit number expressed in base $K$, whose base-10 representation gives the \emph{compression}:
\begin{equation}
z_i^{*} = \sum_{j=1}^p K^{p-j} y_{ij}, \qquad i=1,\ldots,n.
\label{eq:zstar}
\end{equation}

The mapping in \eqref{eq:zstar} is injective, as it ensures that $\mathbf{y}_{i} \neq \mathbf{y}_h \Rightarrow z_{i}^{*} \neq z_{h}^{*} $, $i,h = 1, \ldots, n, i\neq h$. Therefore, it guarantees that distinct statistical units in the original $\mathbf{Y}$ data are retained as distinct in the compressed data vector $\mathbf{z}^*=(z_1^*,\ldots, z_i^*,\ldots, z_n^*)$. It thus projects $\bfY$ into a one-dimensional space through a scalable, deterministic, and easily implemented transformation. 

In exact arithmetic the mapping is reversible, but the weights $K^{p-j}$ grow exponentially in $p$, hence for moderate to large $p$ the transformation becomes numerically infeasible, and reconstruction degrades as precision is lost \citep{shannon:1948}. 
To put this in perspective, IEEE 754 double-precision arithmetic employs a 53-bit significand and can therefore represent all integers exactly up to \(2^{53}\). Using \eqref{eq:zstar}, even a relatively small-dimensional problem with \(p=20\) variables and \(K=7\) can produce values of \(z_i^*\) that exceed this threshold, potentially causing distinct integer values to be represented identically. For example, consider: $\bm{y}_1=(6,6,\ldots,6,6)$ and $\bm{y}_2=(6,6,\ldots,6,5)$. Using \eqref{eq:zstar}, these correspond to $z_1^*=7^{20}-1$ and $z_2^*=7^{20}-7$, respectively. Thus, although \(z_1^*\neq z_2^*\) in exact arithmetic, both values exceed \(2^{53}\), and they differ by only \(6\). Consequently, when stored using standard double-precision arithmetic, the two distinct integers can be rounded to the same floating-point value, illustrating how the theoretically injective mapping can lose its one-to-one property in finite-precision computation. This motivates an alternative that preserves injectivity while remaining numerically stable.

We define the \emph{scaled} base-10 compression of $\mathbf{y}_i$ as
\begin{equation}
z_i = \sum_{j=1}^p K^{\frac{p-j}{pK}} y_{ij}, \qquad i=1,\ldots,n.
\label{eq:z}
\end{equation}
The coefficients in \eqref{eq:z} are uniformly bounded in $[1,K^{1/K})$, with $\sup_{K\geq2}K^{1/K}=3^{1/3}$, and consequently $z_i \in \mathbb{R}^+$. Unlike \eqref{eq:zstar}, the proposed compression remains numerically stable for large $p$ and arbitrary $K$, and it retains injectivity.

\begin{proposition}
\label{prop:1}
The mapping $\mathbf{y}_i \mapsto z_i$ defined in \eqref{eq:z} is injective.
\end{proposition}
The proof is given in Appendix A. The scaled base-10 compression in \eqref{eq:z} is injective,  ensuring that $\mathbf{y}_i \neq \mathbf{y}_h \implies z_i \neq z_h$, hence distinct observations are never collapsed to the same value, making \eqref{eq:z} suitable for dimension reduction and clustering.

\subsection{Compression to \texorpdfstring{$q$}{q} dimensions}

For large $p$, a one-dimensional compression may be too coarse, and representation in $q > 1$ dimensions may be preferable and preserve additional information for clustering. 
To obtain a $q$-dimensional representation, the variables in $\mathbf{Y}$ are partitioned into $q \ll p$ disjoint blocks $B_s$ of sizes $p_s$, $s = 1, \ldots, q$, with $p_s \approx \lfloor p/q \rfloor$. The corresponding submatrices $\mathbf{Y}_s$ are then compressed independently as
\begin{equation}
z_{is}
=
\sum_{j \in B_s}
K^{\frac{p_s-j}{p_sK}} y_{ij},
\qquad
i=1,\ldots,n,\quad s=1,\ldots,q.
\label{eq:zblock}
\end{equation}
Stacking the compressed variables as $\mathbf{z}_i = (z_{i1},\ldots,z_{iq})^\top$ yields the compressed data matrix $\mathbf{Z}\in\mathbb{R}^{n\times q}$.
This blockwise construction preserves the same weighted-sum structure used in the one-dimensional compression, while allowing the information in $\mathbf{Y}$ to be distributed across several compressed coordinates. The resulting representation can therefore be used both for visualization ($q=2,3)$ and for clustering.

\begin{proposition}
\label{distanze}
Let $\mathbf y_i$,$\mathbf y_h\in\{0,\ldots,K-1\}^p$, and let $\mathbf z_i,\mathbf z_h\in\mathbb R^q$ be their compressed representations defined by \eqref{eq:zblock}. Define the distances:
\[
d_{ih} = ||\mathbf y_i-\mathbf y_h||_1,  \qquad
d_{ih}^{*} = ||\mathbf z_i-\mathbf z_h||_2 .
\]
Then
\[
d_{ih}^{*}  \leq K^{\frac{1}{K}} d_{ih} \leq 3^{\frac{1}{3}} d_{ih}.
\]

\end{proposition}
Proof of Proposition \ref{distanze} is provided in Appendix A. The bound shows that the mapping cannot inflate pairwise distances by more than a fixed factor, hence observations that are close in the original space remain close after compression. The compressed representation therefore largely preserves the local distance geometry of the original data $\bfY$, and the example below illustrates that neighborhood structure is indeed broadly retained in practice.

\subsubsection{Illustrative example}
To illustrate how the proposed compression represents binary vectors with different patterns of zeros and ones, consider six binary row vectors of length $p$ (even) that differ in the number and the positions of their nonzero entries: $a = \mathbf{0}_p$, $b = \mathbf{1}_p$, $c = (\mathbf{0}_{p/2}, \mathbf{1}_{p/2})$, $d = (\mathbf{1}_{p/2}, \mathbf{0}_{p/2})$, $e = (\mathbf{0}_{p-2}, 1, 1)$, and $f = (1, 1, \mathbf{0}_{p-2})$. 

\begin{figure}

\centering{
\includegraphics[width=5in]{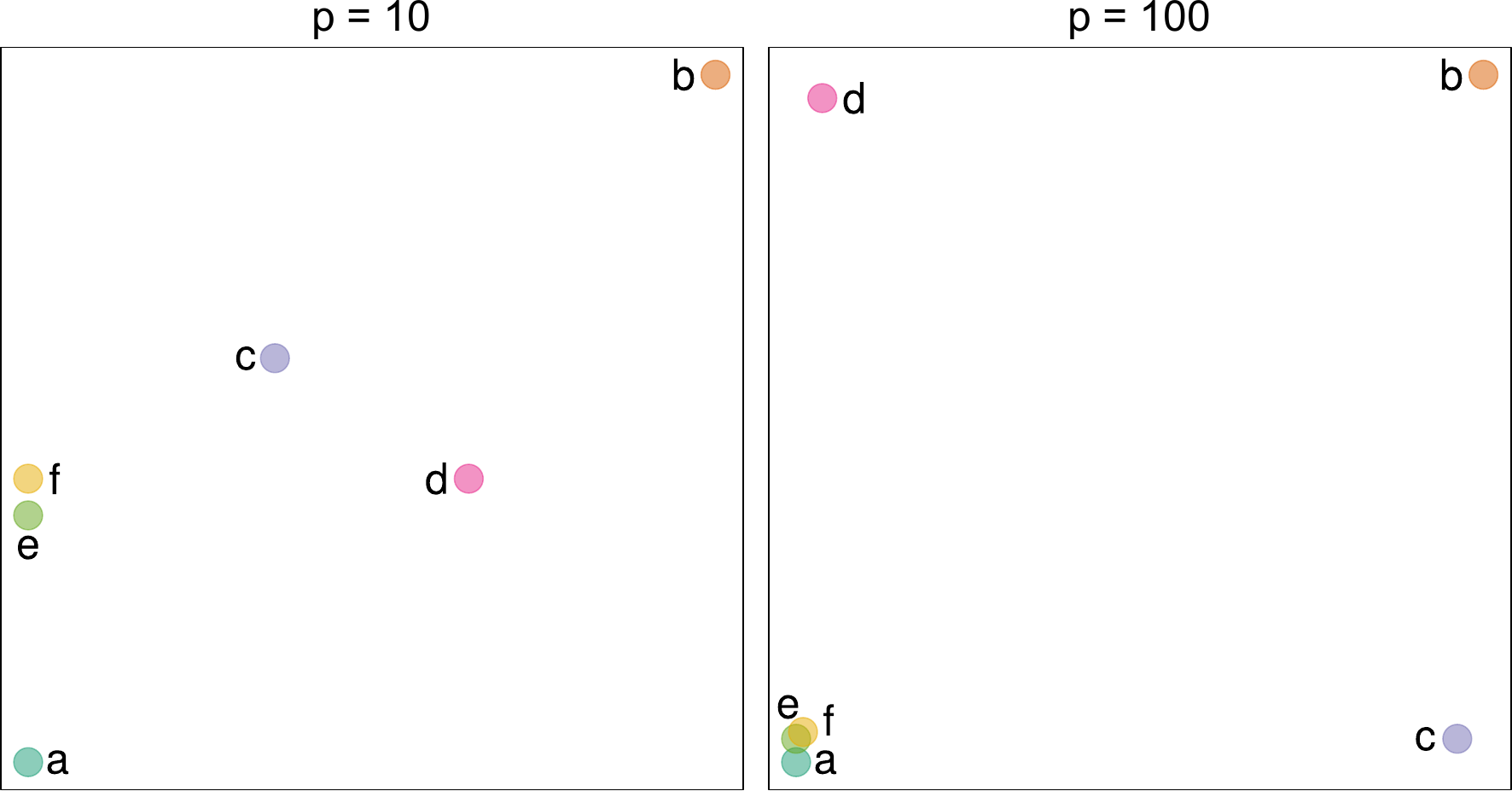}
}

\caption{\label{fig:esempio_comp} Compressed representations of $a$, $b$, $c$, $d$, $e$, $f$ for $p=10$ and $p=100$.}

\end{figure}%

Figure~\ref{fig:esempio_comp} displays their two-dimensional compressed representations for $p=10$ and $p=100$. Since the coordinates in \eqref{eq:zblock} depend on both $p$ and the positions of the nonzero entries, the relative positions of the six vectors in the compressed space vary with $p$. Larger $p$ accentuates the separation between vectors with substantially different patterns, such as $b$, $c$, and $d$, while vectors differing in only a few entries, such as $a$, $e$, and $f$, remain close.

To compare distance structures, we compute Hamming distances on the original binary vectors and Euclidean distances on their corresponding two-dimensional compressed representations. For each vector, we rank the distances to all others in increasing order. Let $A$ and $B$ collect the row-wise rankings under the Hamming and compressed-space Euclidean distances, respectively:
{\footnotesize
\[
A = \begin{array}{c|cccccc}
 & a & b & c & d & e & f \\
\hline
a & - & 3 & 2 & 2 & 1 & 1 \\
b & 3 & - & 1 & 1 & 2 & 2 \\
c & 2 & 2 & - & 4 & 1 & 3 \\
d & 2 & 2 & 4 & - & 3 & 1 \\
e & 1 & 5 & 3 & 4 & - & 2 \\
f & 1 & 5 & 4 & 3 & 2 & - \\
\end{array}
\qquad
B = \begin{array}{c|cccccc}
 & a & b & c & d & e & f \\
\hline
a & - & 5 & 3 & 4 & 1 & 2 \\
b & 5 & - & 2 & 1 & 4 & 3 \\
c & 3 & 4 & - & 5 & 1 & 2 \\
d & 4 & 3 & 5 & - & 2 & 1 \\
e & 2 & 5 & 4 & 3 & - & 1 \\
f & 2 & 5 & 4 & 3 & 1 & - \\
\end{array}
\]
}

Under the Hamming distance, distinct configurations can be equidistant from a reference vector. For example, $c$ and $d$ are equidistant from $a$, as are $e$ and $f$, because the Hamming distance counts only the number of differing positions, not where they occur. The compression instead weights positions differently, hence the Euclidean distances break some of these ties, yielding rankings with no ties in this example. At the same time, the broad structure is retained: the farthest vector from each reference is the same under both rankings, and nearest-neighbor relationships are broadly consistent. In this example, the compression thus refines distance relationships through positional information while preserving the main neighborhood structure of the original data.

Similar considerations apply to count data with $K>2$. There, dissimilarities such as the Bray--Curtis dissimilarity \citep{chen:2025} summarize differences in the original high-dimensional space, whereas the proposed compression provides a low-dimensional representation that preserves broad separation patterns while reducing ties and redundancies induced by discrete observations.

\subsubsection{Approximate distribution of the compressed variables}

Each $z_{is}$ is a linear combination of $p_s$ discrete random variables with deterministic coefficients. The following remark gives its approximate distribution.

\begin{remark}
\label{rem:dist}
If the variables in $\mathbf{Y}_s$, $s=1,\dots, q$, are independent, each with finite mean $\mu_{j,s}$ and variance $\sigma_{j,s}^2$, then, under the Lindeberg condition,
\begin{equation}
z_{is} \overset{\text{approx}}{\sim} 
\mathcal{N} \left(
\sum_{j \in B_s} K^{\frac{p_s-j}{p_sK}}\mu_{j,s},
\sum_{j \in B_s} K^{\frac{2(p_s-j)}{p_sK}}\sigma_{j,s}^2
\right).
\label{eq:zdistribution}
\end{equation}
\end{remark}

Although the original variables are discrete, each compressed variable is a weighted sum of many discrete variables and may therefore behave approximately as a continuous Gaussian variable when the block size is sufficiently large. The applicability of the Central Limit Theorem is further supported by the fact that the coefficients $K^{\frac{p_s-j}{p_sK}}$ in \eqref{eq:zblock} remain uniformly bounded and vary smoothly, rather than growing exponentially as in \eqref{eq:zstar}. Under the independence assumptions of Remark~\ref{rem:dist}, this motivates the use of standard clustering methods in the compressed space. More generally, if the original observations arise from a finite mixture of discrete distributions, \eqref{eq:zdistribution} suggests that the compressed data can retain a lower-dimensional representation of the original mixture structure, as discussed next.

\subsection{Clustering in the compressed data space}
\label{sec:clust_compression_space}

Suppose that the rows of $\mathbf{Y}$ arise from a finite mixture with $G$ components, each component corresponding to a cluster. Let $c_{ig}$ be the latent cluster membership indicator, with $c_{ig}=1$ if observation $i$ belongs to cluster $g$ and $c_{ig}=0$ otherwise. Conditionally on cluster membership, $\mathbf{y}_i \mid (c_{ig}=1) \sim p_g(\bm{\theta}_g)$, where $p_g(\bm{\theta}_g)$ is a component-specific discrete distribution with parameter vector $\bm{\theta}_g$. We denote by $\bm{\mu}_g = \mathbb{E}[\mathbf{y}_i \mid c_{ig}=1]$, $\bm{\sigma}_g^2 = \mathbb{V}\text{ar}[\mathbf{y}_i \mid c_{ig}=1]$ the corresponding component-specific marginal mean and variance parameters, $\bm{\mu}_g=(\mu_{1g},\ldots, \mu_{jg}, \ldots \mu_{pg})$, $\bm{\sigma}^2_g=(\sigma_{1g}^2,\ldots, \sigma_{jg}^2, \ldots \sigma_{pg}^2)$.

Under the conditions of Remark~\ref{rem:dist}, the compressed variables can be interpreted as arising from a lower-dimensional version of the original mixture. Conditionally on cluster membership,
\begin{equation}
z_{is}\mid(c_{ig}=1)  \overset{\text{approx}}{\thicksim} \mathcal{N} \left(
\sum_{j \in B_s}K^{\frac{p_s-j}{p_sK}}\mu_{gj,s}\,,\,
\sum_{j \in B_s}K^{\frac{2(p_s-j)}{p_sK}}\sigma_{gj,s}^2
\right),
\label{eq:zclust}
\end{equation}
where $\mu_{gj,s}$ and $\sigma_{gj,s}^2$ are the mean and variance of the $j$th variable in block $s$ for component $g$, with $s=1,\ldots,q$. Thus a finite mixture of discrete distributions in the original space maps to an approximate finite mixture in the compressed space. This approximation motivates model-based clustering  of $\mathbf{Z}$. Gaussian mixture models are natural candidates, since \eqref{eq:zclust} suggests an approximate Gaussian form for the compressed variables within each cluster. 

The compression is also compatible with distance-based methods such as k-means, because distinct observations and distinct cluster centroids remain distinct
after compression. The following proposition formalizes this latter point for $q=1$; the argument extends to $q > 1$ blockwise. Define the cluster centroids in the original space as $\mathbf{m}_g = \frac{1}{n_g} \sum_{i}c_{ig} \mathbf{y}_{i}$, where $n_g = \sum_i c_{ig}$; the associated compressed data cluster centroids are $u_g = \frac{1}{n_g} \sum_{i}c_{ig} z_{i}$. 

\begin{proposition}
\label{prop:medie}
If the cluster centroids ${\mathbf{m}}_1,\ldots, \mathbf{m}_g,\ldots,{\mathbf{m}}_G$ in the original space are distinct, i.e. any two centroids differ in at least one element, then the associated compressed centroids ${u}_1,\ldots, u_g,\ldots,u_G$ are also distinct.
\end{proposition}
Proof of Proposition~\ref{prop:medie} is in Appendix A. Distinct cluster centroids in the original data space thus remain distinct after compression. Hence, when the clustering structure is driven primarily by differences in cluster location, the compression preserves the separation between cluster positions. For $q>1$, the same argument applies blockwise: if two original centroids differ in at least one variable, then they differ in the compressed coordinate of the block containing it.

The next remark quantifies how much separation is retained after compression, showing how the signal-to-noise ratio governing cluster separation scales with the block size.

\begin{remark}
\label{rem:medie_conv}
Consider the standardized difference between two compressed coordinates of two observations belonging to clusters $g$ and $l$. Under independence and uniformly bounded fourth moments, as $p_s \rightarrow \infty$ with $\frac{p}{n} \rightarrow c$, the Lindeberg Central Limit Theorem gives:
\[
\frac{\left(z_{is} - z_{hs}\right)  }{\sqrt{\mathbb{V}\mathrm{ar}\left(z_{is} \mid c_{ig} = 1\right) + \mathbb{V}\mathrm{ar}\left(z_{hs} \mid c_{il} = 1\right)}} \xrightarrow[]{d} \mathcal{N} \left(\delta_{gls}, 1 \right),
\]
where 
\[
\delta_{gls} = 
\frac{ \sum_{j \in B_s} K^{\frac{p_s-j}{Kp_s}} \left( m_{gjs}  -  m_{ljs} \right)}{\sqrt{\mathbb{V}\mathrm{ar}\left(z_{is} \mid c_{ig} = 1\right) + \mathbb{V}\mathrm{ar}\left(z_{hs} \mid c_{il} = 1\right)}} ,
\]
for $s = 1, \dots, q$ and $g,l = 1, \dots, G$. 
 
\end{remark}
Hence, the cluster separation signal $\delta_{gls}$ increases with the block size $p_s$, provided differences between cluster means accumulate across variables. This suggests that clustering performance in the compressed data may improve as dimension increases,  in contrast to distance-based methods that deteriorate under the curse of dimensionality. The remark also highlights a bias-variance trade-off in the choice of $q$: $q=1$ yields the strongest Gaussian approximation through maximal aggregation, while $q>1$ spreads the clustering signal across coordinates, which can help methods that exploit multidimensional geometry.

Once $\mathbf{Z}$ is obtained, clustering can be performed directly in the compressed space. In principle, any clustering method can be applied to $\mathbf{Z}$. In this work, we focus on two popular complementary choices, aligned with \eqref{eq:zclust} and Proposition~\ref{prop:medie}: k-means and Gaussian mixture models.

K-means \citep{bishop:2006} is a simple, efficient centroid-based method, appropriate when the compressed clusters are well separated and approximately spherical. Under this approach, selection of the number of clusters is implemented by maximizing the average silhouette width over a prespecified range of candidate values \citep{rousseeuw:1987}. Other criteria could also be used for this purpose, such as the gap statistic \citep{tibshirani:2001}. We use the silhouette for its simplicity and demonstrated performance in a variety of contexts \citep{arbelaitz:2013,batool:2021}.

Model-based clustering based on finite Gaussian mixtures \citep{bouveyron:2019} is instead motivated by the approximate distribution in \eqref{eq:zclust}. We fit Gaussian mixtures using the R package \texttt{mclust} \citep{mclust}, which supports a range of covariance parameterizations. Model selection and selection of the number of clusters are implemented using the Bayesian information criterion \citep[BIC,][]{fraley:2002}. Compared with k-means, this approach can accommodate clusters with unequal volume, orientation, or dispersion in the compressed space.

The two methods use the compression differently: k-means relies on the separation between cluster centroids in the compressed space, while Gaussian mixtures exploit the approximate Gaussian mixture structure it induces. The simulation study in Section 3 evaluates both.

\subsection{Variable ordering}
\label{sec:var_order}

The coefficients in \eqref{eq:z} and \eqref{eq:zblock} depend on the ordering of the variables in the data matrix. Injectivity and centroid separation are invariant to this ordering, but the dispersion of the compressed variables is not.

Let $\pi_s$ be a permutation of $\{1,\ldots,p_s\}$, with $\pi_s(j)$ the index of the variable assigned to position $j$ in block $s$. The ordering of the variables affects the coefficients assigned to them in the compressed representation. Writing $a_s = K^{\frac{2}{p_sK}}$, assuming that variables are independent within each block, from \eqref{eq:zdistribution} the variance of the compressed variable associated with block $s$ under ordering $\pi_s$ is:
\begin{equation}
\sigma_{\pi_s,s}^2
=
\sum_{j \in B_s}
a_s^{p_s-j}
\sigma^2_{s,\pi_s(j)}
=
\sum_{j \in B_s}
K^{\frac{2(p_s-j)}{p_sK}}
\sigma^2_{s,\pi_s(j)}.
\label{eq:block_variance_ordering}
\end{equation}
Equivalently, this expression can be written in Horner form \citep{borwein:1995} as: $\sigma_{\pi_s,s}^2 =
\sigma^2_{s,\pi_s(p_s)}
+
a_s\left(
\sigma^2_{s,\pi_s(p_s-1)}
+\cdots+
a_s\left(
\sigma^2_{s,\pi_s(2)}
+
a_s\sigma^2_{s,\pi_s(1)}
\right)
\right)$. 

Expression \eqref{eq:block_variance_ordering} shows that different permutations of the variable indices within a block can lead to different levels of variability in the compressed representation. Since $a_s \geq 1$, the weights 
$a_s^{p_s-1}, a_s^{p_s-2}, \ldots, a_s, 1$ are nonincreasing. By the rearrangement inequality,  \eqref{eq:block_variance_ordering} is maximized 
by assigning the largest variance to the largest weight, the second largest variance to the second largest weight, and so on. Hence the maximizing ordering satisfies $\sigma^2_{s,\pi_s(1)}
\geq
\sigma^2_{s,\pi_s(2)}
\geq
\cdots
\geq
\sigma^2_{s,\pi_s(p_s)}$.
We therefore order the variables by decreasing sample variance before compression. This increases the contribution of high-variance variables, which is consistent with the expectation of greater heterogeneity when a clustering structure is present.

This argument assumes within-block independence. Under correlation, \eqref{eq:block_variance_ordering} acquires covariance terms and the variance-based ordering need not be optimal. We empirically assess the sensitivity of the compression to variable ordering in terms of both the stability of the compressed representation and clustering performance. For details and results we refer to Section \ref{sec:sim:2_vo}.

\section{Simulation study}
\label{sec:sim}

We conduct a comprehensive simulation study to evaluate the clustering performance of the proposed compression method for high-dimensional discrete data, defined in \eqref{eq:z} and hereafter denoted by DC. Performance is assessed across a range of sample sizes, dimensions, cluster structures, noise levels, and sparsity patterns. These settings mimic realistic high-dimensional discrete-data applications. Unless otherwise stated, DC reduces $\bfY \in \mathbb{N}^{n\times p}$ to two or five dimensions.

The proposed DC approach is compared with several widely used dimensionality reduction and clustering procedures. Details on implementations are provided below.

\begin{itemize}
    \item PCA: Principal component analysis \citep{jolliffe:2002}, implemented using the \texttt{prcomp} function in R. The number of principal components is selected so that they explain at least $80\%$ of the total variance.
    
    \item t-SNE: t-distributed stochastic neighbor embedding \citep{maaten:2008}, implemented using the \texttt{Rtsne} R package \citep{rtsne}. The data are projected onto two dimensions using the default perplexity value.
    
    \item MDS: Multidimensional scaling \citep{borg:2005}, implemented using the \texttt{cmdscale} function in R with the Bray--Curtis dissimilarity \citep{chen:2025}. The data are projected onto two dimensions.
    
    \item SC: Spectral clustering \citep{luxburg:2007}, implemented using the \texttt{specc} function from the \texttt{kernlab} R package \citep{kernlab1,kernlab2}. The number of clusters is selected using the average silhouette criterion.
    
    \item KM: Standard k-means clustering \citep{bishop:2006}, implemented using the \texttt{kmeans} function in R. The number of clusters is also selected using the average silhouette criterion.
\end{itemize}

For DC, PCA, t-SNE, and MDS, clustering is performed on the reduced representation using either Gaussian mixture modeling via \texttt{mclust} (MC) or k-means. We denote the combinations accordingly: MDS+MC is MDS followed by Gaussian mixture clustering, whereas DC+KM(5) is compression to five dimensions followed by k-means, and so on. This design separates the effects of dimension reduction and clustering. Performance is measured using the adjusted Rand index \citep[ARI,][]{hubert:1985}.

Computational times for DC and the competing methods are reported in Section \ref{sec:sim:1}. Here, we instead emphasize the differences in their theoretical computational complexities. The proposed DC has complexity linear in both $n$ and $p$, namely $O(np)$, whereas the competing methods have higher computational complexity: $O\left(np\min(n,p)\right)$ for PCA, $O\left(n^2p+n\log n\right)$ for t-SNE, $O\left(n^2p+n^3\right)$ for MDS, $O\left(n^2p+n^3\right)$ for SC, and $O(npG)$ for KM.

\subsection{Scenario 1: Independent count data}
\label{sec:sim:1}

We first consider independent Poisson data $y_{ij}\mid z_i=g\sim\text{Pois}(\lambda_g)$, with $n\in \{50,100,200 \}$, $p \in \{100,500,1000,2000,5000 \}$, and $G \in \{2,5 \}$. For $G=2$, $\bm{\lambda}=(3,4)$; for $G=5$, $\bm{\lambda}=(1,2,3,4,5)$. Variables are independent within clusters and 100 datasets are generated for each configuration.

\begin{figure}[t]

\centering{

\includegraphics[width=5.5in]{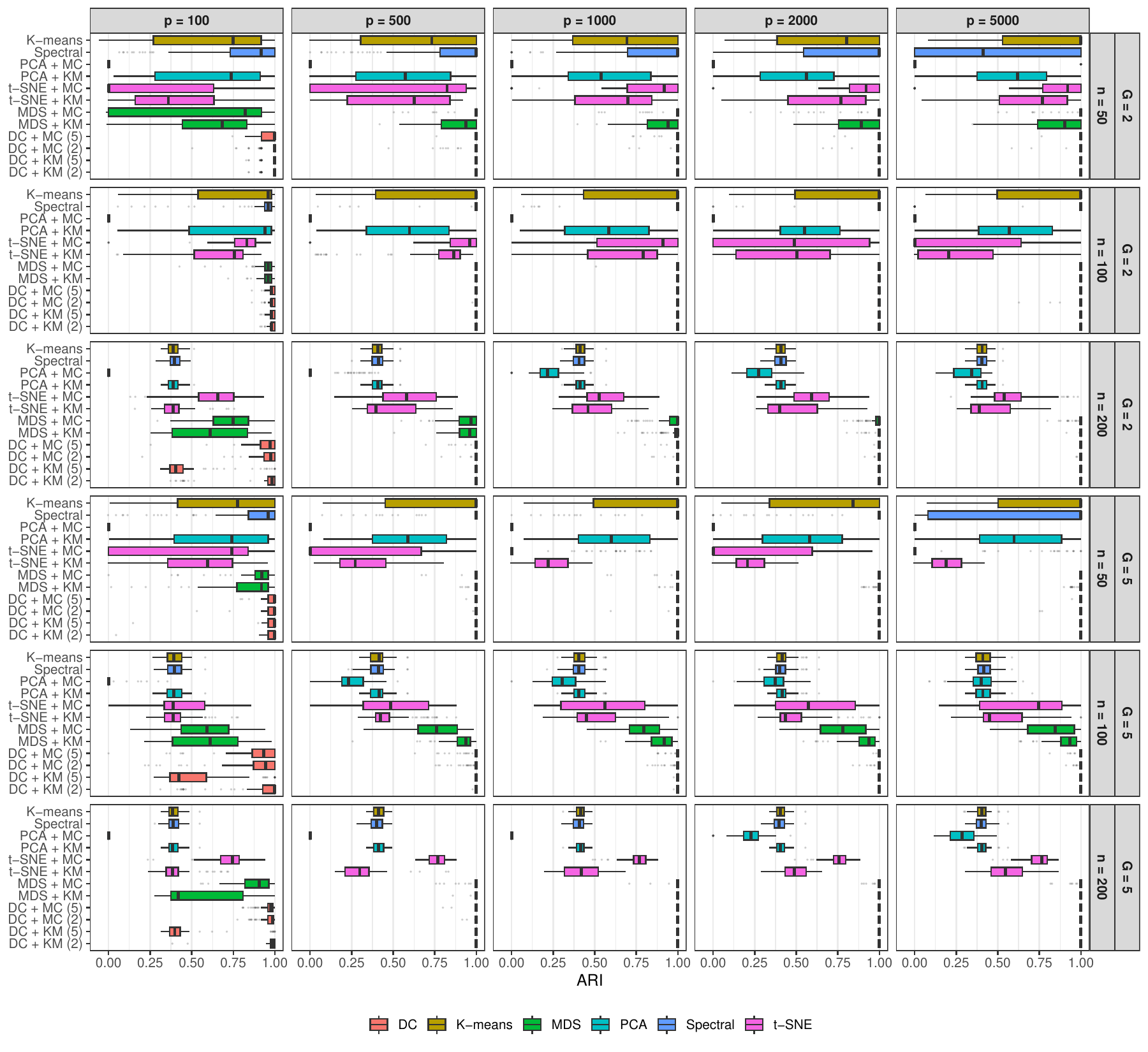}

}

\caption{\label{fig:ari_sim1} Simulations, Scenario 1. ARI values across methods.}

\end{figure}%
Figure~\ref{fig:ari_sim1} shows that DC-based clustering achieves near-perfect recovery across almost all settings, with little sensitivity to the number of compressed dimensions and similar results for k-means and \texttt{mclust}. Among competing approaches, SC and MDS perform best but exhibit greater variability. Representative projections in Appendix B show clear cluster separation under DC, whereas separation is weaker for PCA, t-SNE, and MDS. Note that for PCA, the median number of retained components is 65 (range 28--152).
\begin{table}[tb!]
\centering
\begin{tabular}{lrrrrrr}
\toprule
\em Method comparison & \em Min. & \em 1st Qu. & \em Median & \em Mean & \em 3rd Qu. & \em Max. \\
\midrule
MDS over DC(2)    &  1.71 &  5.37 &  11.02 &  14.65 &  21.75 &   70.87 \\
MDS over DC(5)    &  1.83 &  5.82 &  12.43 &  16.19 &  24.86 &   72.51 \\
PCA over DC(2)    &  1.86 &  8.44 &  16.65 &  20.12 &  28.93 &   84.61 \\
PCA over DC(5)    &  2.74 &  9.38 &  18.85 &  22.46 &  32.70 &   91.02 \\
t-SNE over DC(2)  & 12.27 & 50.61 &  94.57 & 181.68 & 225.59 & 1168.59 \\
t-SNE over DC(5)  & 13.07 & 56.64 & 109.97 & 193.66 & 242.64 & 1224.28 \\
\bottomrule
\end{tabular}
  \caption{ \label{tab:times_red_sim1} Simulations, Scenario 1. Ratios of computational times for competing dimensionality-reduction methods relative to DC.}
\end{table}
DC also yelds substantial computational savings. Fore example, relative to DC(2), MDS is approximately 14 times slower on average, PCA 20 times slower, and t-SNE 181 times slower (Table~\ref{tab:times_red_sim1}), reflecting the underlying computational complexities. Including clustering time, DC yields average speedups of approximately 9-fold over spectral clustering and 50-fold over k-means. See Appendix B for additional results.

\subsection{Scenario 2: Correlated count data}
\label{sec:sim:2}

To assess robustness to violations of the independence assumption in Remark~\ref{rem:dist}, we generate cluster-specific multivariate Poisson data using \texttt{simstudy} \citep{simstudy}. We fix $G=3$ with $\bm{\lambda}=(1,2,3)$ and consider moderate correlations $(-0.5,0.5)$ and high correlations $(-1,1)$, with $n\in\{50,100,200\}$ and $p\in\{100,500\}$.

\begin{figure}

\centering{

\includegraphics[width=5.5in]{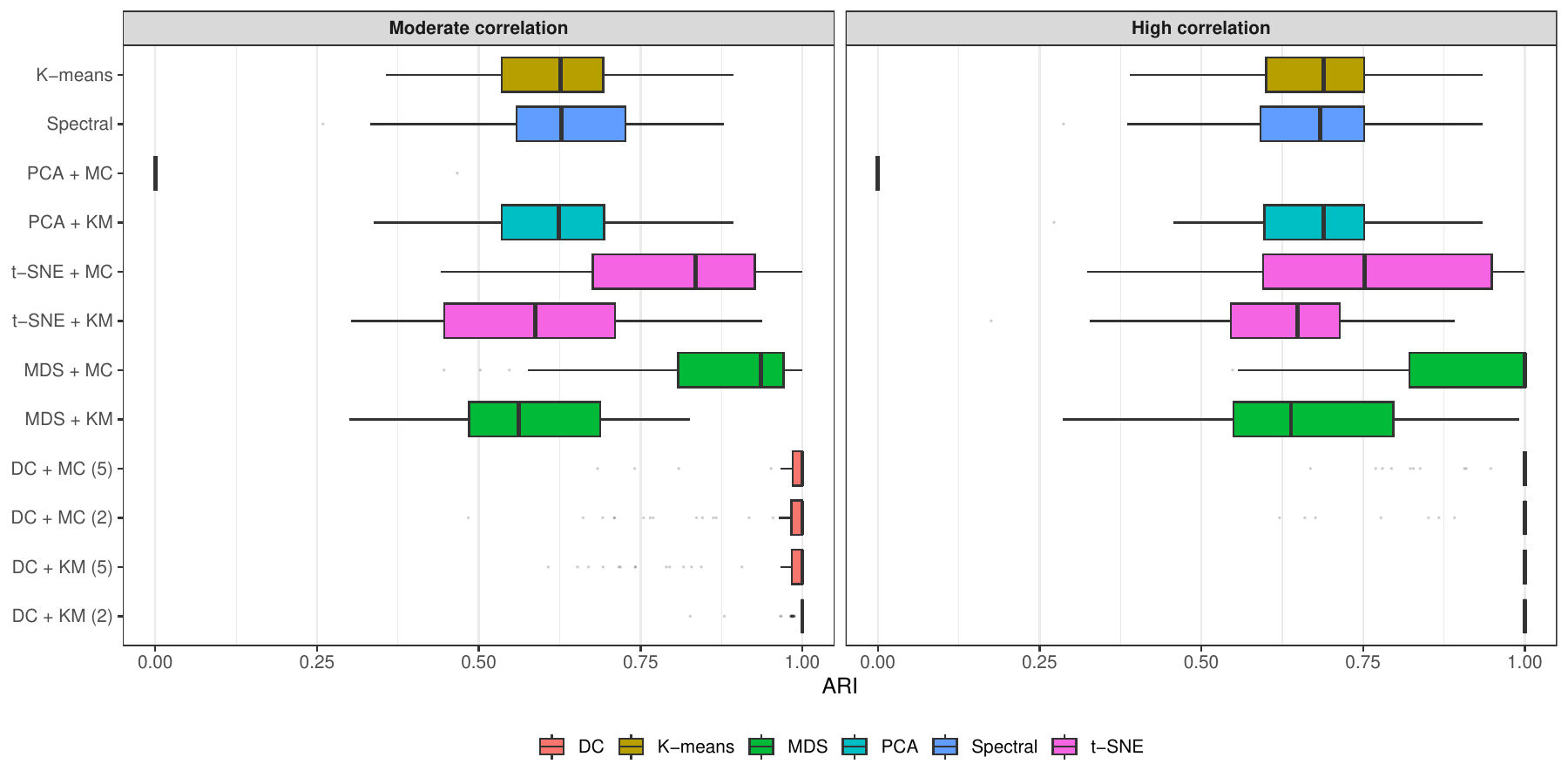}

}

\caption{\label{fig:sim2_paper} Simulations, Scenario 2. ARI values across methods ($n=100$, $p=500$).}

\end{figure}%

Figure~\ref{fig:sim2_paper} shows that DC-based clustering continues to achieve near-perfect recovery under both correlation regimes, with k-means and \texttt{mclust} giving comparable results. Competing methods are less consistent, with MDS and t-SNE providing the strongest alternatives. Similar patterns are observed across all $(n,p)$ combinations; see Appendix B.  A representative projection there show that t-SNE
and MDS capture part of the cluster structure, whereas PCA (retaining a median of 39 components, range 21–173), tends to place observations from different clusters close together.

\subsubsection{Variable ordering}
\label{sec:sim:2_vo}
As anticipated in Section \ref{sec:var_order}, we designed a simulation study to empirically assess the sensitivity of the compression to variable ordering, to evaluate the stability of the compressed representation and the clustering performance.

Data are generated with $n=100$ under the settings of the previous Scenario 2 (Section \ref{sec:sim:2}).
To assess the stability of the compressed space with respect to variable ordering, we compare the compressed spaces obtained under the proposed variance-based ordering with those obtained under 100 random permutations of the variables. This comparison is performed for both two- and five-dimensional compressed spaces. Similarity is measured using the pairwise Procrustes correlation, computed with the \texttt{protest} function in the \texttt{vegan} R package. Across simulated datasets, the correlations range from 0.94 to 1, with median 0.97, for $p=100$, and from 0.98 to 1, with median 0.99, for $p=500$. These results indicate that the compressed representation is nearly invariant to variable ordering.
\begin{figure}[t]
\centering
\includegraphics[scale=.54]{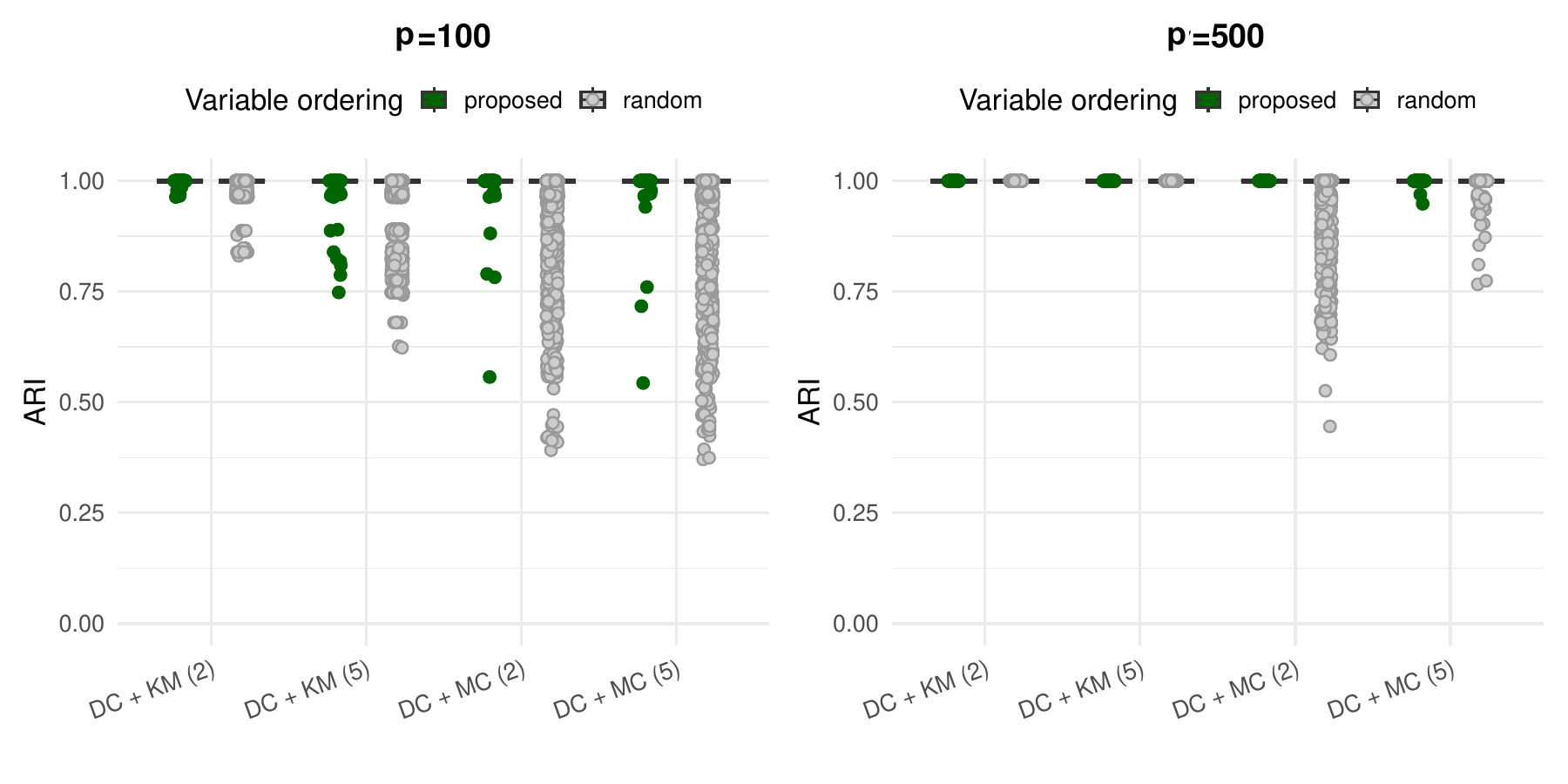}
\caption{\label{fig:sim_order_var_ari} Variable ordering comparison: ARI for true vs. estimated class memberships.}
\end{figure}
We next examine the impact of variable ordering on clustering performance by comparing adjusted Rand index \citep[ARI,][]{hubert:1985} values. The results are summarized in Figure~\ref{fig:sim_order_var_ari}. For $p=100$, the proposed ordering yields higher ARI values than random orderings in $95\%$ of cases for \texttt{mclust} clustering and in $98\%$ of cases for k-means clustering. For $p=500$, the ARI values obtained under the proposed ordering are greater than or equal to those obtained under random orderings in nearly all cases, with the only exception being \texttt{mclust} clustering on the five-dimensional compressed space, where this occurs in $98\%$ of cases.
Overall, while the compressed data representations are largely stable under permutations of the variable order, ordering the variables within each block by decreasing sample variance provides consistent gains in clustering performance.
\subsection{Scenario 3: Binary data}
\label{sec:sim:3}

We next consider binary observations from cluster-specific Bernoulli distributions, with $G=3$, $n\in\{50,100,200\}$, and $p\in\{100,500\}$. Two probability configurations are examined: $\bm{\pi}=(0.3,0.5,0.7)$, corresponding to well-separated clusters, and $\bm{\pi}=(0.3,0.4,0.5)$, corresponding to weaker separation. In this scenario, the Hamming distance is used in input to MDS. 
\begin{figure}

\centering{

\includegraphics[width=5.5in]{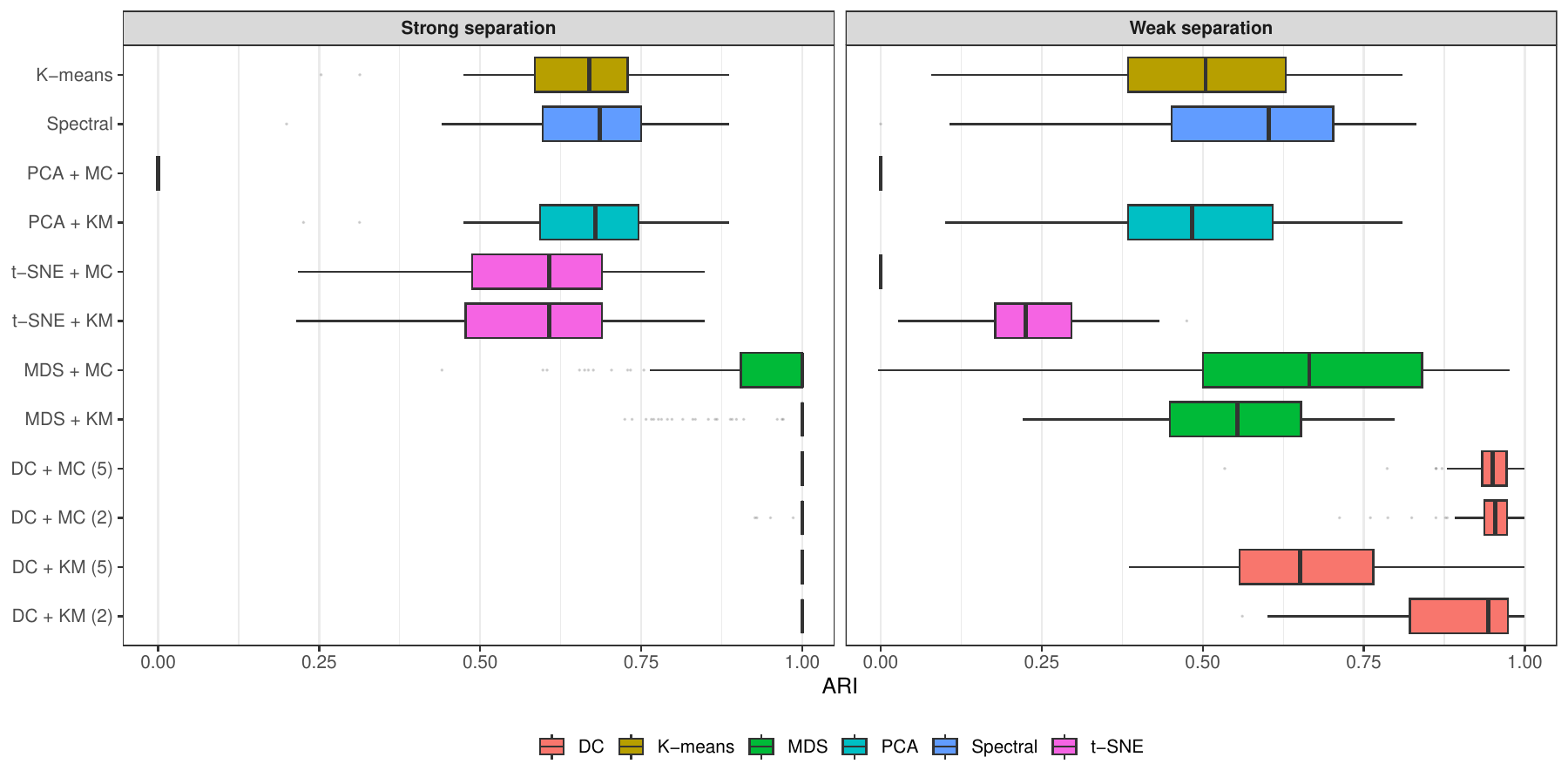}

}

\caption{\label{fig:sim3_paper}Simulations, Scenario 3. ARI values across methods ($n=100$, $p=500$).}

\end{figure}%

Figure~\ref{fig:sim3_paper} shows that DC-based clustering consistently outperforms competitors. As expected, performance improves with stronger cluster separation, although all methods deteriorate when separation is weak, though the decline is less pronounced as dimension increases. Full results are reported in Appendix B.

\subsection{Scenario 4: Count data with noise}
\label{sec:sim:4}

To evaluate robustness to irrelevant features, data are generated as in Scenario~1 with $G=3$, $\bm{\lambda}=(1,2,3)$, adding a fraction $r\in\{0.2,0.5,0.8\}$ of noise variables generated independently from $\text{Pois}(1.5)$ and unrelated to cluster membership.
\begin{figure}

\centering{

\includegraphics[width=5.5in]{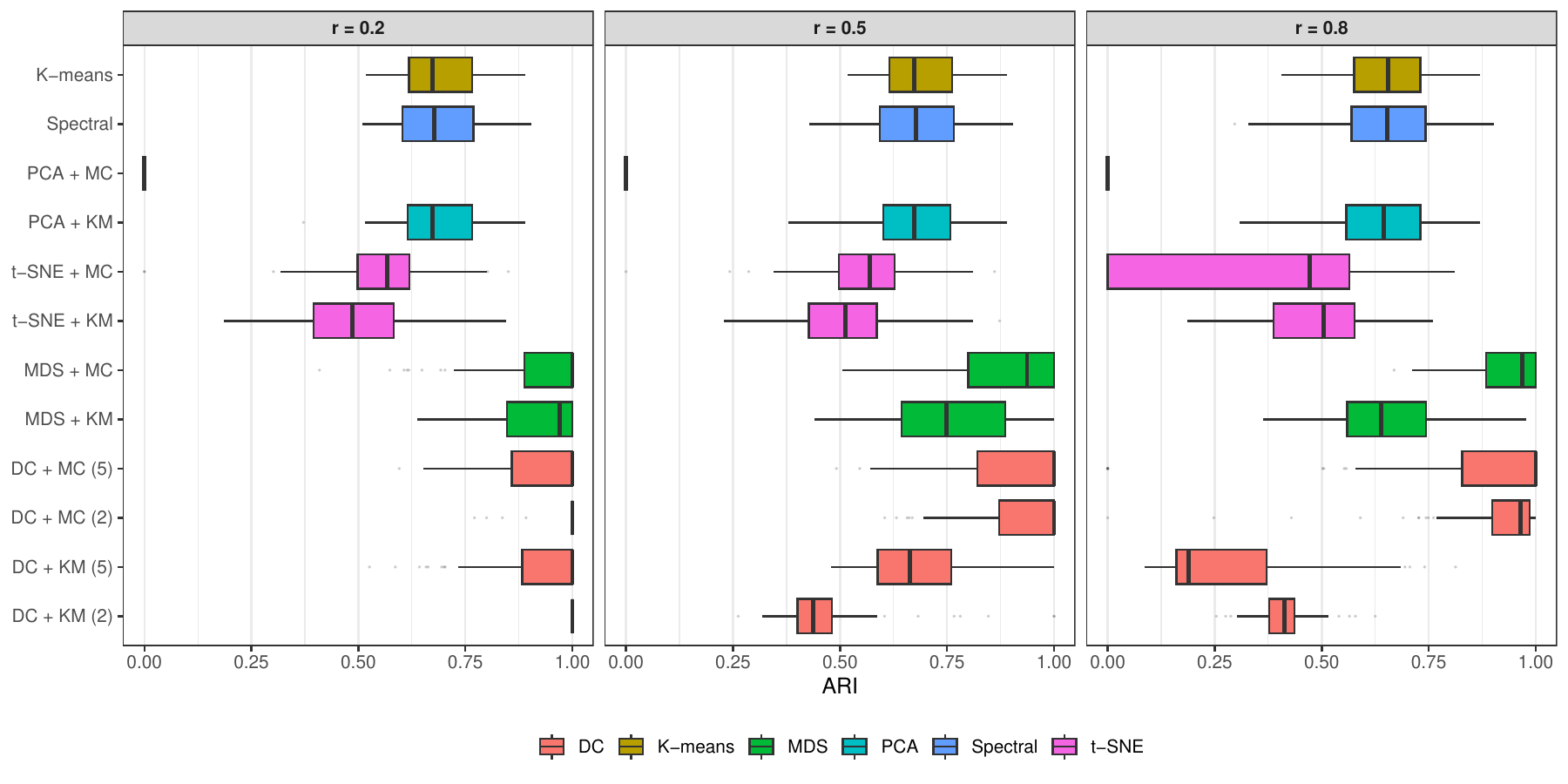}

}

\caption{\label{fig:sim4_paper} Simulations, Scenario 4. ARI values across methods ($n=100$, $p=500$).}

\end{figure}%

Figure~\ref{fig:sim4_paper} shows that DC-based clustering remains the strongest performer, although accuracy declines as the noise proportion increases. In this setting, DC+MC consistently outperforms DC+KM, likely because noise inflates within-cluster variability and induces unequal dispersions that Gaussian mixtures accommodate better than k-means. MDS clustering is the closest competitor, while PCA performs poorly. As noise increases, PCA retains more components (median 70; range 21--173), suggesting that additional dimensions largely capture noise. Representative projections are in Appendix B.

\subsection{Scenario 5: Sparse data}
\label{sec:sim:5}

Finally, we assess performance under sparsity. Data are generated as in Scenario~2 with $G=3$ and $\bm{\lambda}=(1,2,3)$, after which a proportion $r\in\{0.4,0.7,0.9\}$ of entries is randomly set to zero, inducing unstructured sparsity unrelated to cluster membership.
\begin{figure}

\centering{

\includegraphics[width=5.5in]{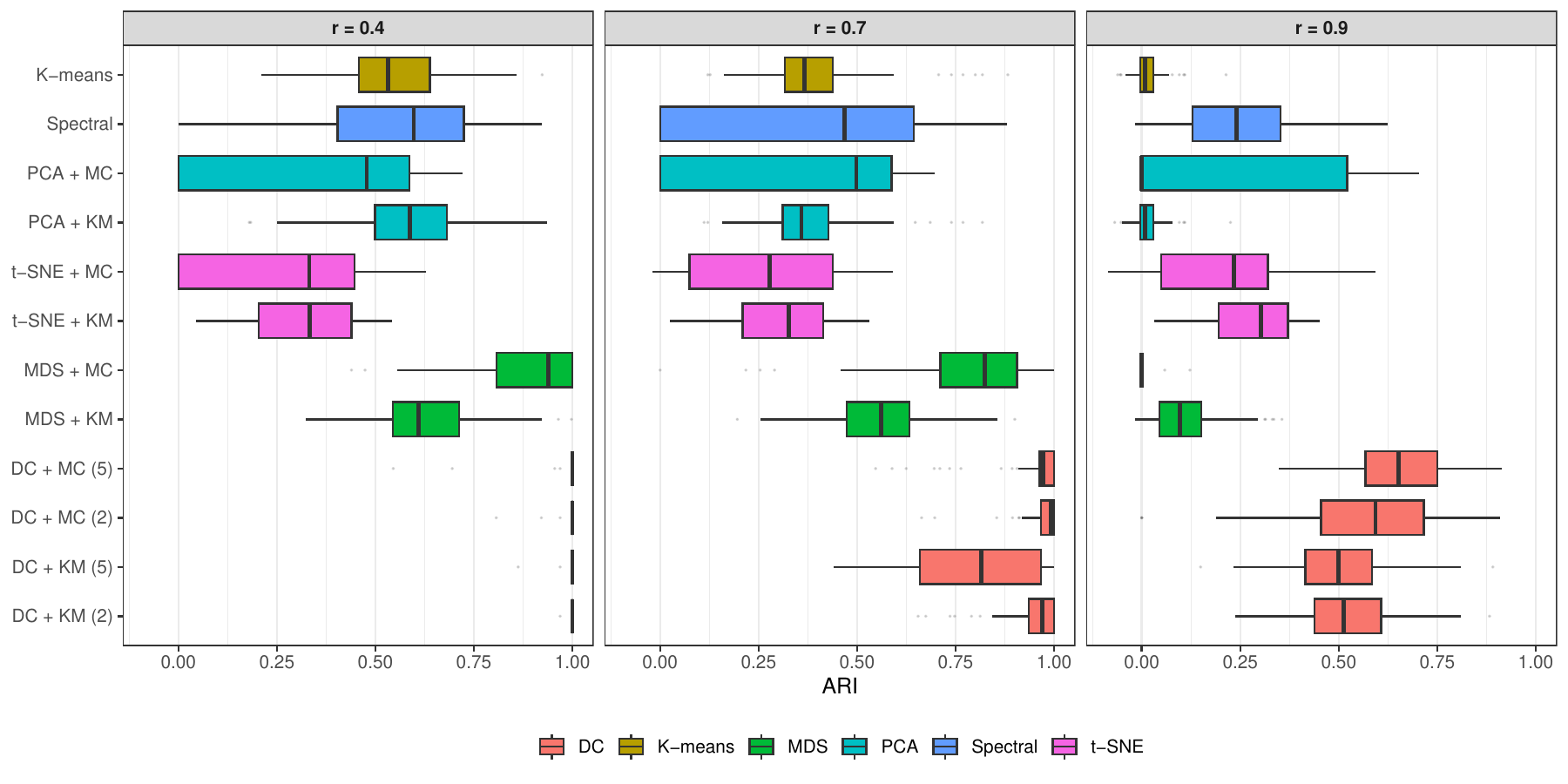}

}

\caption{\label{fig:sim5_paper} Simulations, Scenario 5. ARI values across methods ($n=100$, $p=500$).}

\end{figure}%

Figure~\ref{fig:sim5_paper} shows that DC-based clustering consistently outperforms competing methods, with performance declining for all approaches as sparsity increases. As in Scenario~4, \texttt{mclust} generally outperforms k-means in the compressed space, likely because sparsity induces uneven dispersion and non-spherical clusters. MDS is again the strongest alternative but is less stable across sparsity levels and dimensions. A representative case ($p=100$, $n=100$, $r=0.7$) in Appendix B shows DC still yields partial separation and \texttt{mclust} recovers the clusters with moderate accuracy (ARI = 0.49).

\section{Illustrative applications}
\label{sec:application}

We illustrate the proposed data compression clustering method through two real-world applications. The first concerns a microbiome dataset for individuals in populations with different diets \citep{Schnorr:2014}; the second concerns historical voting data of the United Nations General Assembly \citep{voeten:2013}. In both cases we compare DC clustering with MDS-based clustering, the best-performing competitor in the simulation studies.

\subsection{Schnorr microbiome data}

The dataset collected by \cite{Schnorr:2014} consists of fecal microbiome samples from 16 Italian adults living in urban environments and 27 Hadza hunter-gatherers from Tanzania, for a total of $n = 43$ samples. The Hadza are one of the few remaining populations maintaining a traditional diet with minimal exposure to processed foods, and differences between the two groups are therefore expected.  The microbiome profiles comprise $p = 4707$ operational taxonomic unit (OTU) count variables.

\begin{table}[bt!]
\centering
\begin{tabular}{ lrr  }
 \toprule
& \multicolumn{2}{c}{Clusters}\\
 \midrule
 & \em 1 & \em 2\\
 \midrule
 Italy & 5 & 11 \\
 Tanzania & 27 & 0 \\
\bottomrule
\end{tabular}
 \caption{ \label{tab:schnorr} Schnorr microbiome data. Known population labels cross-tabulated against the two-cluster partition obtained by \texttt{mclust} on the DC representation with $q = 2$.}
\end{table}

\begin{figure}[t]
\centering{
\includegraphics[scale=0.6]{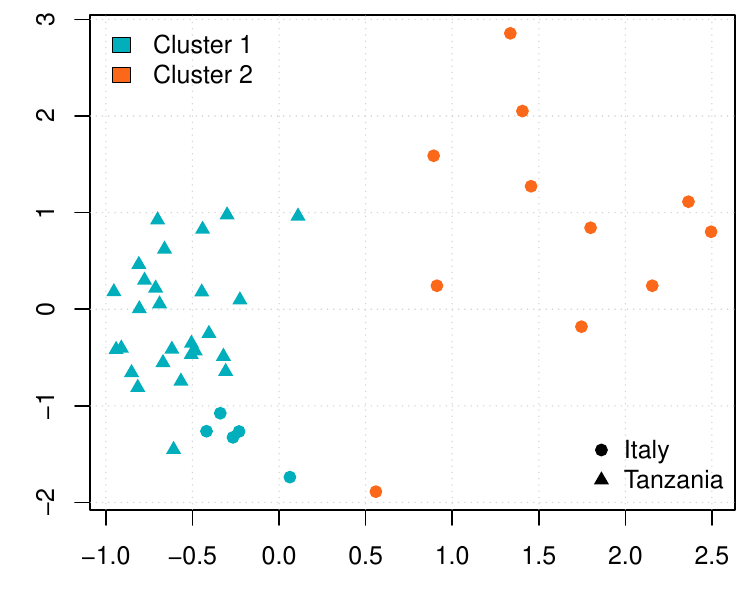}
}
\caption{\label{fig:schnorr_class} Schnorr microbiome data. Samples in the
two-dimensional compressed space obtained by DC with $q = 2$. The axes are the
two standardized compressed coordinates. Symbols indicate the population of
origin, while colors indicate clusters estimated by \texttt{mclust}.}
\end{figure}%

We apply DC and MDS with $q = 2$ and cluster the resulting representations using \texttt{mclust}, with $G \in \{1, \dots, 10\}$ and all covariance parameterizations, letting BIC select the model. For DC, BIC selects $G = 2$, and the resulting partition attains an ARI of 0.58 against the known population labels. Table~\ref{tab:schnorr} and Figure~\ref{fig:schnorr_class} compare the two classifications. The 27 Hadza samples form a compact group and are all allocated to Cluster~1, consistent with \cite{Schnorr:2014}, who report lower within-group variability among the Hadza than among the Italians. The Italian samples are more dispersed and divide between the two clusters, with five assigned to the Hadza-dominated group. This larger dispersion may reflect more heterogeneous diets in an urban population, compared with the traditional diet shared by the Hadza. The recovered partition appears to reflect this heterogeneity, with a subset of Italian samples lying at the edge of the Hadza group in the compressed space. By contrast, MDS-based clustering selects $G = 7$ and yields an ARI of $-0.03$, showing no recovery of the underlying dietary distinction.

To assess the sensitivity of this result to the choice of $q$, we adopt a simple ensemble framework \citep{babu:2025,Casa:2021,fern:2003}. We apply DC and MDS for embedding dimensions $q \in \{2, \dots, 10\}$ and cluster each representation with \texttt{mclust} as above. The nine partitions are aggregated into a co-occurrence matrix recording the proportion of times each pair of samples is allocated to the same cluster. Figure~\ref{fig:schnorr_heatmaps} shows the corresponding heatmaps. DC yields a stable partition into two main groups across the range of $q$, whereas MDS produces a granular structure. Across the ensemble, the median ARI against the Italian--Tanzanian classification is 0.50 for DC and $-0.02$ for MDS.

\begin{figure}[t]
    \centering    
    \begin{subfigure}[b]{0.45\textwidth}
    \centering{

\includegraphics[width=3in]{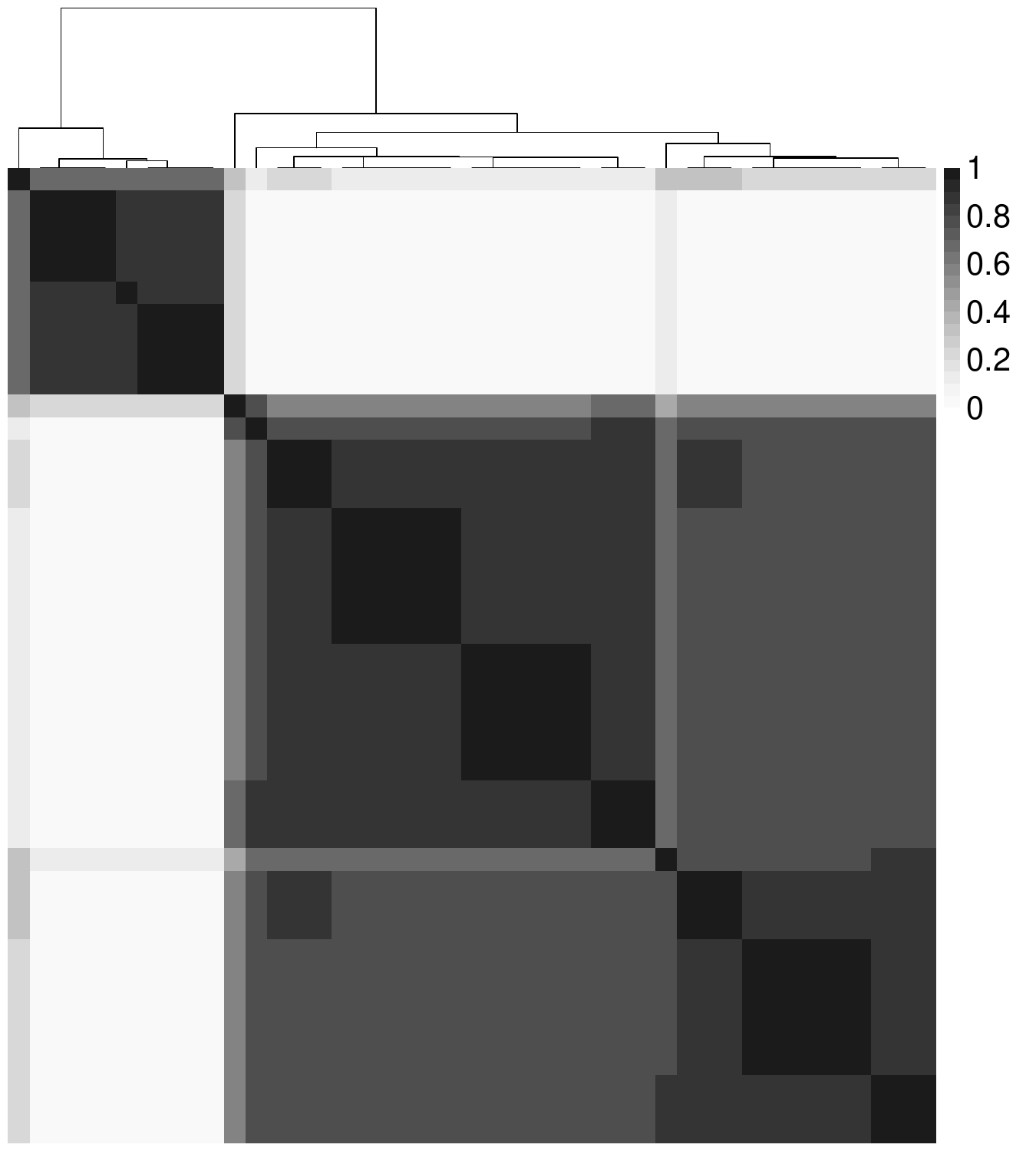}

}
    \caption{DC}
        \label{fig:un2s}
    \end{subfigure}
    \begin{subfigure}[b]{0.45\textwidth}
        \centering{

\includegraphics[width=3in]{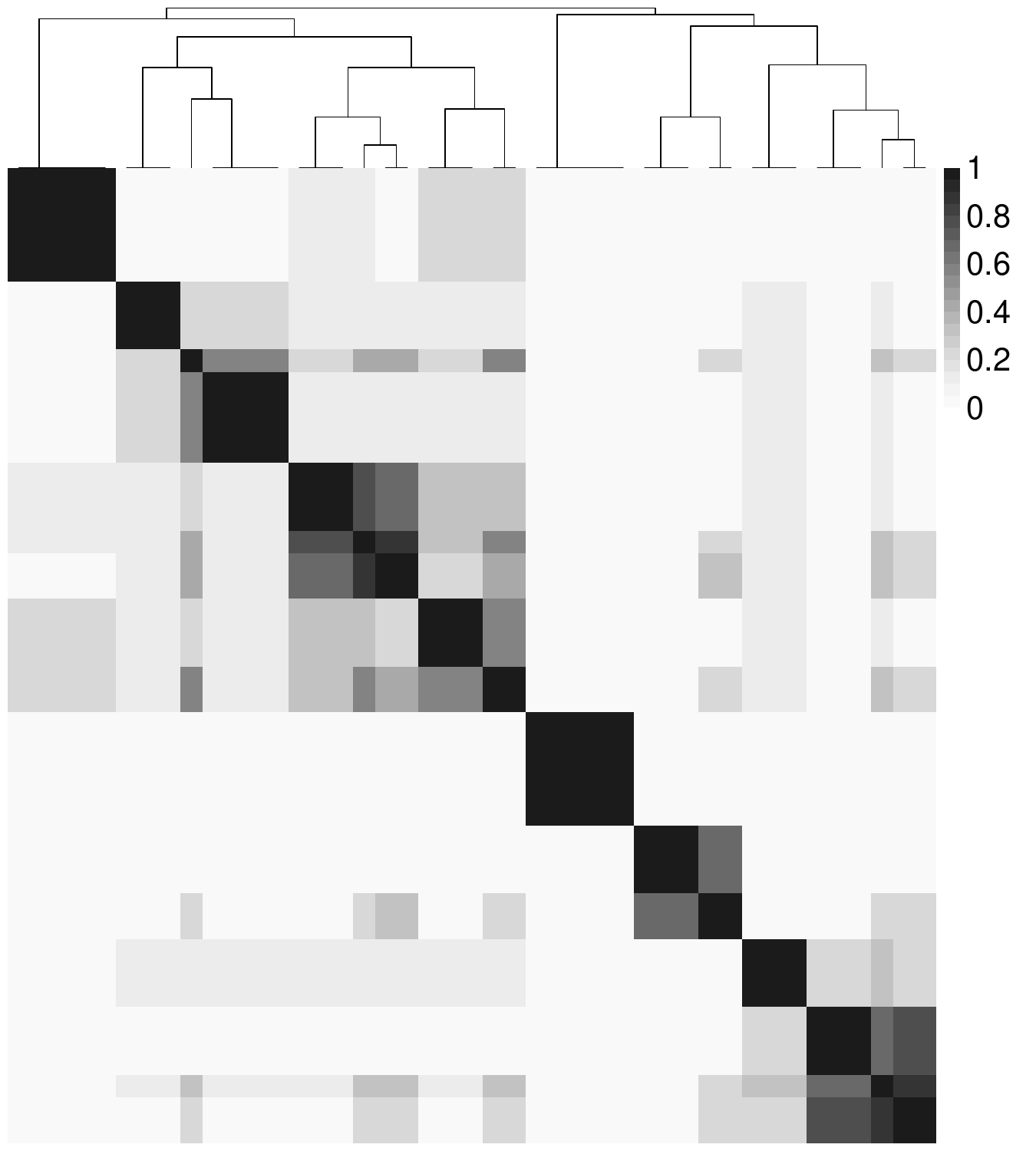}

}
        \caption{MDS}
        \label{fig:schnorr_heatmaps1}
    \end{subfigure}
    \caption{\label{fig:schnorr_heatmaps} Schnorr microbiome data. Heatmaps of the co-occurrence matrices obtained across embedding dimensions $q \in \{2, \dots, 10\}$ for the DC and MDS representations. Each cell reports the proportion of partitions in which two samples are allocated to the same cluster; dendrograms are obtained by average-linkage hierarchical clustering.}
    \end{figure}

Overall, the compressed representation recovers the distinction between the Hadza and Italian samples, and does so stably across embedding dimensions, whereas the MDS representation does not. The residual disagreement with the known labels is concentrated in the Italian samples, which is consistent with the greater dietary variability of that group rather than with a failure to separate the populations. These findings agree with \cite{Shi:2022}, who report weak performance of distance-based methods in recovering the class structure of this dataset.

\subsection{United Nations General Assembly voting data}
The R package \texttt{unvotes} \citep{unvotes} records voting in the United Nations General Assembly between 1946 and 2020 \citep{voeten:2013}, covering 200 countries across 6202 roll calls. Because some states ceased to exist while others were created or joined the United Nations after 1946, votes are missing for countries that were not members at the time of a given roll call. We therefore encode the data with $K = 3$ categories: ``1'' and ``2'' denote ``no'' and ``yes'', while ``0'' denotes a missing vote, so that absent countries do not contribute to the compression for those roll calls. We fix $q = 5$, a choice that performed well across all simulation scenarios, and fit \texttt{mclust} with 2 to 15 clusters to the representations obtained from DC and from MDS on Hamming distances. Since no reference classification exists for these data, the two solutions are compared with each other rather than against a ground truth.

DC yields a solution with seven clusters that appear to reflect broad geographical, political, and economic features. Cluster 2 groups long-term high-income democracies with histories of stable governance and membership of institutions such as the OECD, the EU, and NATO-related structures; most of Western Europe falls here. Cluster 7 resembles Cluster 2 but is largely post-socialist, comprising states that underwent transitions following the Cold War and the collapse of the Soviet Union (e.g.\ North Macedonia). Clusters 3 and 4 mostly contain developing and emerging economies: Cluster 3 includes regional powers (e.g.\ Mexico, Saudi Arabia) and middle-income states (e.g.\ Colombia), concentrated in South America and the Middle East, while Cluster 4 contains many African, Asian, and Caribbean countries. Cluster 6 consists mainly of small island and geographically peripheral states, several with tourism-based economies (e.g.\ Bahamas, St.\ Vincent and the Grenadines). Clusters 1 and 5 are more heterogeneous. Cluster 1 is predominantly African, together with some former European states, the United States, Iraq, and Israel; Cluster 5 includes central Asian states (e.g.\ Turkmenistan, Kazakhstan), small island states (e.g.\ Tuvalu), and countries with distinctive voting positions such as North Korea and Switzerland. These two groups admit no single interpretation, and their composition should be read as descriptive.

Figure~\ref{fig:un_paper} examines within-cluster cohesion across the seven clusters by showing the yearly frequency of ``yes'' votes. For each cluster, the $95\%$ confidence interval is represented by a shaded gray band, while representative countries are highlighted by colored curves. Overall, the figure suggests a high degree of cohesion in voting behavior within clusters over time, although the width of the confidence bands varies across clusters, indicating differences in their degree of homogeneity.

MDS-based clustering also yields seven clusters, in moderate agreement with the DC solution (ARI = 0.38). DC Clusters 2, 5, and 7 correspond closely to two MDS clusters, with almost identical membership, and a core subset of DC Cluster 1 remains grouped under MDS. The countries in the remaining DC clusters are redistributed differently across the MDS solution. Full cluster compositions are reported in Appendix C.

Overall, the compressed representation recovers a partition of the voting data that is cohesive and broadly interpretable in terms of geopolitical alignment. Several clusters correspond to recognizable groupings (long-term high-income democracies, post-socialist states, emerging economies, small island states) recovered directly from the roll-call matrix without any external covariate information. The yearly voting profiles in Figure~\ref{fig:un_paper} indicate that this structure is temporally coherent: countries grouped together on the basis of the compressed coordinates also tend to vote similarly over time, even though the compression treats the roll calls as exchangeable and uses no chronological
information.

\begin{figure}[t]
\centering{

\includegraphics[width=6.3in]{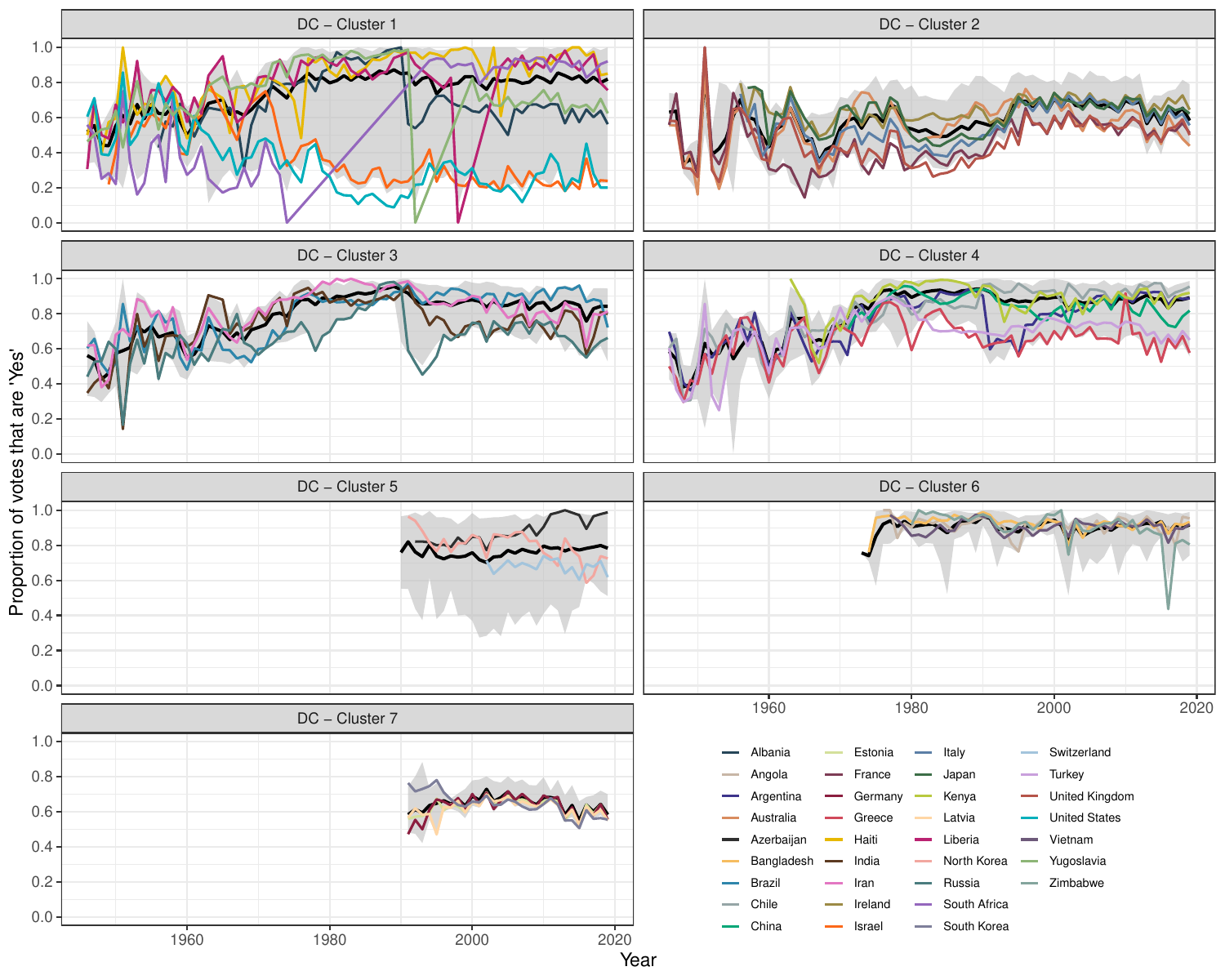}

}
      \caption{\label{fig:un_paper} United Nations voting data. Yearly ``yes'' frequency for the DC-based partition. Gray bands give the $95\%$ confidence interval for the average yearly ``yes'' frequency within each cluster, and colored curves a selection of representative countries per cluster.}
    \end{figure}

\section{Discussion}
\label{sec:discussion}

We introduced a fast, deterministic compression framework for dimension reduction and clustering of high-dimensional discrete data. The method maps discrete observations to a low-dimensional continuous representation through weighted aggregations of the original variables, producing a numerically stable embedding for binary and count data. It avoids the pairwise dissimilarity matrices, eigendecompositions, stochastic optimization, and high-dimensional probabilistic modeling required by many existing approaches, yielding substantial computational gains while retaining meaningful cluster structure. The supporting theory links the original discrete space to its compressed representation and motivates both model-based and distance-based clustering after compression.

An extensive simulation study supported the validity of the proposed method. Across independent and correlated count data, binary observations, sparse settings, and scenarios with many irrelevant variables, clustering on the compressed representation accurately recovered the underlying partition, and performance remained strong even when the independence assumption underlying the Gaussian approximation was violated. Computationally, the method was consistently faster than PCA, multidimensional scaling, and t-SNE.

Some aspects of the data compression framework merit further consideration. The theory primarily addresses clusters distinguished by location: Proposition~\ref{prop:medie} guarantees preservation of centroid separation, and Remark~\ref{rem:medie_conv} shows how mean differences accumulate under compression. Clusters differing mainly in variance, dispersion, or dependence structure may therefore be harder to distinguish, though in many discrete-data settings this is mitigated because location and dispersion are intrinsically linked, as in Poisson and multinomial models. A second limitation concerns irrelevant variables: although simulations indicate robustness to moderate noise, variables unrelated to clustering still contribute to the compressed coordinates and can reduce the signal-to-noise ratio. As in other high-dimensional clustering methods, feature screening or variable selection prior to compression \citep{krazanowski:2009,fop:2018} could address this, particularly in ultra-high-dimensional settings. 

Several extensions merit investigation. Alternative block constructions, including data-adaptive or hierarchical partitions, may improve performance under heterogeneous dependence structures. While Gaussian mixture models and k-means are natural choices in the compressed space, nonparametric mixtures, density-based methods, and graph-based procedures could also be employed. Because the compression is independent of the downstream clustering algorithm, it may serve as a general-purpose preprocessing step for a broad range of clustering methodologies. A further open problem is inference after clustering: many applications require inference on cluster differences once a partition has been estimated, and recent work on selective and post-clustering inference \citep{enjalbert:2025, chen:2023} offers a foundation for extending the framework to hypothesis testing and uncertainty quantification in the compressed space.

More broadly, this work shows that deterministic aggregation can provide an effective alternative to conventional dimension-reduction techniques for clustering discrete data. By combining a simple positional encoding with bounded coefficients, the proposed method offers scalability, theoretical tractability, and satisfactory empirical performance. As high-dimensional discrete datasets become more common, compression-based approaches may offer a useful foundation for future methodological development.

\bibliographystyle{apalike}
\bibliography{bibliography}

@article{chen:2025,
    author = {Guanhua Chen and Xinyue Wang and Qiang Sun and Zheng-Zheng Tang},
    title = {Multidimensional scaling improves distance-based clustering for microbiome data},
    journal = {Bioinformatics},
    volume = {41},
    number = {2},
    pages = {},
    year = {2025},
    month = {01}
}

@book{borwein:1995,
  author    = {Peter Borwein and Tam{\'a}s Erd{\'e}lyi},
  title     = {Polynomials and Polynomial Inequalities},
  series    = {Graduate Texts in Mathematics},
  volume    = {161},
  publisher = {Springer},
  address   = {New York},
  year      = {1995},
  isbn      = {978-0-387-94509-5}
}

@article{shannon:1948,
  author  = {Claude E. Shannon},
  title   = {A Mathematical Theory of Communication},
  journal = {Bell System Technical Journal},
  volume  = {27},
  pages   = {379--423, 623--656},
  year    = {1948}
}

@book{lang:2002,
  author    = {Serge Lang},
  title     = {Algebra},
  edition   = {3},
  publisher = {Springer},
  address   = {New York},
  year      = {2002}
}

@book{borg:2005,
  author    = {Ingwer Borg and Patrick J. F. Groenen},
  title     = {Modern Multidimensional Scaling: Theory and Applications},
  edition   = {2},
  publisher = {Springer},
  address   = {New York},
  year      = {2005},
  isbn      = {978-0-387-25150-9},
  doi       = {10.1007/0-387-28981-X}
}

@article{maaten:2008,
  author  = {Laurens van der Maaten and Geoffrey Hinton},
  title   = {Visualizing Data Using {t-SNE}},
  journal = {Journal of Machine Learning Research},
  volume  = {9},
  pages   = {2579--2605},
  year    = {2008},
  issn    = {1532-4435}
}

@article{luxburg:2007,
  author  = {Ulrike von Luxburg},
  title   = {A Tutorial on Spectral Clustering},
  journal = {Statistics and Computing},
  volume  = {17},
  number  = {4},
  pages   = {395--416},
  year    = {2007},
  doi     = {10.1007/s11222-007-0033-z}
}

@book{jolliffe:2002,
  author    = {Ian T. Jolliffe},
  title     = {Principal Component Analysis},
  edition   = {2},
  publisher = {Springer},
  address   = {New York},
  year      = {2002},
  isbn      = {978-0-387-95442-4}
}

@book{bishop:2006,
  author    = {Christopher M. Bishop},
  title     = {Pattern Recognition and Machine Learning},
  publisher = {Springer},
  address   = {New York},
  year      = {2006},
  isbn      = {978-0-387-31073-2}
}

@article{Casa:2021,
  author  = {Alessandro Casa and Luca Scrucca and Giovanna Menardi},
  title   = {{Better than the best? Answers via model ensemble in density-based clustering}},
  journal = {Advances in Data Analysis and Classification},
  volume  = {15},
  number  = {},
  pages   = {599-–623},
  year    = {2021},
  doi     = {10.1007/s11634-020-00423-6}
}

@article{fern:2003,
  author  = {Xiaoli Zhang Fern and Carla E. Brodley},
  title   = {Random projection for high dimensional data clustering: a cluster ensemble approach},
  journal = {Proceedings of the 20th international conference on machine learning},
  volume  = {},
  number  = {},
  pages   = {186-–193},
  year    = {2003},
  doi     = {}
}

@article{Schnorr:2014,
  author  = {Stephanie L. Schnorr and Marco Candela and Simone Rampelli and Manuela Centanni and Clarissa Consolandi and Giulia Basaglia and Silvia Turroni and Elena Biagi and Clelia Peano and Marco Severgnini and Jessica Fiori and Roberto Gotti and Gianluca De Bellis and Donata Luiselli and Patrizia Brigidi and Audax Mabulla and Frank Marlowe and Amanda G. Henry and Alyssa N. Crittenden },
  title   = {{Gut microbiome of the Hadza hunter-gatherers}},
  journal = {Nature Communications},
  volume  = {5},
  number  = {3654},
  pages   = {},
  year    = {2014},
  doi     = {10.1038/ncomms4654}
}

@article{Shi:2022,
  author  = {Yushu Shi and Liangliang Zhang 2 and Christine B. Peterson and Kim-Anh Do and Robert R. Jenq},
  title   = {Performance determinants of unsupervised clustering methods for microbiome data},
  journal = {Microbiome},
  volume  = {10},
  number  = {25},
  pages   = {},
  year    = {2022},
  doi     = {10.1186/s40168-021-01199-3}
}

@article{Wang:2025,
  author  = {Changhu Wang and Zihao Chen and Ruibin Xi},
  title   = {Feature screening for clustering analysis of count data with an application to single-cell RNA-sequencing},
  journal = {Annals of Applied Statistics},
  volume  = {19},
  number  = {4},
  pages   = {2738--2758},
  year    = {2025},
  doi     = {10.1214/25-AOAS2102}
}

@Manual{vegan,
    title = {vegan: Community Ecology Package},
    author = {Jari Oksanen and Gavin L. Simpson and F. Guillaume Blanchet and Roeland Kindt and Pierre Legendre and Peter R. Minchin and R.B. O'Hara and Peter Solymos and M. Henry H. Stevens and Eduard Szoecs and Helene Wagner and Matt Barbour and Michael Bedward and Ben Bolker and Daniel Borcard and Tuomas Borman and Gustavo Carvalho and Michael Chirico and Miquel {De Caceres} and Sebastien Durand and Heloisa Beatriz Antoniazi Evangelista and Rich FitzJohn and Michael Friendly and Brendan Furneaux and Geoffrey Hannigan and Mark O. Hill and Leo Lahti and Cameron Martino and Dan McGlinn and Marie-Helene Ouellette and Eduardo {Ribeiro Cunha} and Tyler Smith and Adrian Stier and Cajo J.F. {Ter Braak} and James Weedon},
    year = {2026},
    note = {R package version 2.7-3},
    url = {https://CRAN.R-project.org/package=vegan},
    doi = {10.32614/CRAN.package.vegan},
  }

@article{hubert:1985,
  title={Comparing partitions},
  author={Lawrence Hubert and Phipps Arabie},
  journal={Journal of classification},
  volume={2},
  number={1},
  pages={193--218},
  year={1985},
  publisher={Springer}
}

@Manual{kernlab1,
    title = {kernlab: Kernel-Based Machine Learning Lab},
    author = {Alexandros Karatzoglou and Alex Smola and Kurt Hornik},
    year = {2024},
    note = {R package version 0.9-33},
    url = {https://CRAN.R-project.org/package=kernlab},
    doi = {10.32614/CRAN.package.kernlab},
  }

@Article{kernlab2,
    title = {kernlab -- An {S4} Package for Kernel Methods in {R}},
    author = {Alexandros Karatzoglou and Alex Smola and Kurt Hornik and Achim Zeileis},
    journal = {Journal of Statistical Software},
    year = {2004},
    volume = {11},
    number = {9},
    pages = {1--20},
    doi = {10.18637/jss.v011.i09},
  }

@Manual{rtsne,
    title = {{Rtsne}: T-Distributed Stochastic Neighbor Embedding using Barnes-Hut
Implementation},
    author = {Jesse H. Krijthe},
    year = {2015},
    note = {R package version 0.17},
    url = {https://github.com/jkrijthe/Rtsne},
  }

@Article{simstudy,
    title = {simstudy: Illuminating research methods through data generation},
    author = {Keith Goldfeld and Jacob Wujciak-Jens},
    publisher = {The Open Journal},
    journal = {Journal of Open Source Software},
    year = {2020},
    volume = {5},
    number = {54},
    pages = {2763},
    url = {},
    doi = {10.21105/joss.02763},
  }

@article{arbelaitz:2013,
title = {An extensive comparative study of cluster validity indices},
journal = {Pattern Recognition},
volume = {46},
number = {1},
pages = {243-256},
year = {2013},
issn = {0031-3203},
doi = {https://doi.org/10.1016/j.patcog.2012.07.021},
url = {},
author = {Olatz Arbelaitz and Ibai Gurrutxaga and Javier Muguerza and Jesús M. Pérez and Iñigo Perona}
}

@article{batool:2021,
title = {Clustering with the Average Silhouette Width},
journal = {Computational Statistics \& Data Analysis},
volume = {158},
pages = {107190},
year = {2021},
issn = {0167-9473},
doi = {https://doi.org/10.1016/j.csda.2021.107190},
url = {},
author = {Fatima Batool and Christian Hennig}
}

@Article{lausser:2018,
author={Ludwig Lausser and Florian Schmid and 
 Lyn-Rouven Schirra  and Adalbert F. X. Wilhelm and Hans A. Kestler},
title={Rank-based classifiers for extremely high-dimensional gene expression data},
journal={Advances in Data Analysis and Classification},
year={2018},
month={Dec},
day={01},
volume={12},
number={4},
pages={917-936},
issn={1862-5355},
doi={10.1007/s11634-016-0277-3},
url={}
}

@article{yang:2021,
author = {Congyuan Yang and Carey E. Priebe and Youngser Park and David J. Marchette},
title = {Simultaneous Dimensionality and Complexity Model Selection for Spectral Graph Clustering},
journal = {Journal of Computational and Graphical Statistics},
volume = {30},
number = {2},
pages = {422--441},
year = {2021},
publisher = {Taylor \& Francis},
doi = {10.1080/10618600.2020.1824870}
}

@article{rousseeuw:1987,
  author  = {Peter J. Rousseeuw},
  title   = {Silhouettes: A Graphical Aid to the Interpretation and Validation of Cluster Analysis},
  journal = {Journal of Computational and Applied Mathematics},
  volume  = {20},
  pages   = {53--65},
  year    = {1987},
  doi     = {10.1016/0377-0427(87)90125-7}
}

@article{tibshirani:2001,
    author = {Robert Tibshirani and Guenther Walther and Trevor Hastie},
    title = {Estimating the Number of Clusters in a Data Set Via the Gap Statistic},
    journal = {Journal of the Royal Statistical Society Series B: Statistical Methodology},
    volume = {63},
    number = {2},
    pages = {411-423},
    year = {2001},
    month = {07},
    issn = {1369-7412},
    doi = {10.1111/1467-9868.00293},
    url = {https://academic.oup.com/jrsssb/article-pdf/63/2/411/49590410/jrsssb_63_2_411.pdf},
}

@book{bouveyron:2019, 
place={Cambridge}, series={Cambridge Series in Statistical and Probabilistic Mathematics}, title={Model-Based Clustering and Classification for Data Science: With Applications in R}, publisher={Cambridge University Press}, author={Charles Bouveyron and Gilles Celeux and T. Brendan Murphy and Adrian E. Raftery}, year={2019}, collection={Cambridge Series in Statistical and Probabilistic Mathematics}}

@book{mclust,
    title     = {Model-Based Clustering, Classification, and Density Estimation Using {mclust} in {R}},
    author    = {Luca Scrucca and Chris Fraley and T. Brendan Murphy and Adrian E. Raftery},
    publisher = {Chapman and Hall/CRC},
    isbn      = {978-1032234953},
    doi       = {10.1201/9781003277965},
    year      = {2023}
}

@manual{R,
    title        = {{R}: A Language and Environment for Statistical Computing},
    author       = {R Core Team},
    organization = {R Foundation for Statistical Computing},
    address      = {Vienna, Austria},
    year         = {2026},
    url          = {https://www.R-project.org/}
}

@article{fraley:2002,
author = {Chris Fraley and Adrian E Raftery},
title = {Model-Based Clustering, Discriminant Analysis, and Density Estimation},
journal = {Journal of the American Statistical Association},
volume = {97},
number = {458},
pages = {611--631},
year = {2002},
publisher = {Taylor \& Francis},
doi = {10.1198/016214502760047131},
}

@Article{anderlucci:2022,
author={Laura Anderlucci and Francesca Fortunato 
and Angela Montanari},
title={High-Dimensional Clustering via Random Projections},
journal={Journal of Classification},
year={2022},
month={Mar},
day={01},
volume={39},
number={1},
pages={191-216},
issn={1432-1343},
doi={10.1007/s00357-021-09403-7},
url={}
}

@Article{payne:2025,
author={Andrea Payne and Anjali Silva and Steven J. Rothstein 
and Paul D. McNicholas and  Sanjeena Subedi},
title={Finite mixtures of multivariate Poisson-log normal factor analyzers for clustering count data},
journal={Statistics and Computing},
year={2025},
month={Sep},
day={15},
volume={35},
number={6},
pages={189},
issn={1573-1375},
doi={10.1007/s11222-025-10720-9},
url={}
}

@Article{fang:2023,
author={Yuan Fang Sanjeena Subedi},
title={Clustering microbiome data using mixtures of logistic normal multinomial models},
journal={Scientific Reports},
year={2023},
month={Sep},
day={07},
volume={13},
number={1},
pages={14758},
issn={2045-2322},
doi={10.1038/s41598-023-41318-8},
url={}
}

@article{bouveyron:2014,
title = {{Model-based clustering of high-dimensional data: A review}},
journal = {Computational Statistics \& Data Analysis},
volume = {71},
pages = {52-78},
year = {2014},
issn = {0167-9473},
doi = {https://doi.org/10.1016/j.csda.2012.12.008},
url = {},
author = {Charles Bouveyron and Camille Brunet-Saumard},
}

@article{yao:2025,
  author  = {Dapeng Yao and Fangzheng Xie and Yanxun Xu},
  title   = {{B}ayesian Sparse {G}aussian Mixture Model for Clustering in High Dimensions},
  journal = {Journal of Machine Learning Research},
  year    = {2025},
  volume  = {26},
  number  = {21},
  pages   = {1--50},
  url     = {}
}

@article{zhou:2025,
author = {Yijia Zhou and Kyle A. Gallivan and Adrian Barbu},
title = {{Scalable clustering: Large scale unsupervised learning of Gaussian mixture models with outliers}},
journal = {Journal of Computational and Graphical Statistics},
volume = {34},
number = {3},
pages = {884--895},
year = {2025},
publisher = {Taylor \& Francis},
doi = {10.1080/10618600.2024.2414889}
}

@article{amiri:2018,
author = {Saeid Amiri and Bertrand S. Clarke and Jennifer L. Clarke},
title = {Clustering Categorical Data via Ensembling Dissimilarity Matrices},
journal = {Journal of Computational and Graphical Statistics},
volume = {27},
number = {1},
pages = {195--208},
year = {2018},
publisher = {Taylor \& Francis},
doi = {10.1080/10618600.2017.1305278}
}

@article{wang:nonpar:2025,
author = {Yong Wang and Reza Modarres},
title = {Clustering of high-dimensional observations},
journal = {Journal of Nonparametric Statistics},
volume = {37},
number = {2},
pages = {319--343},
year = {2025},
publisher = {Taylor \& Francis},
doi = {10.1080/10485252.2024.2378904}
}

@article{champon:2026,
author = {Xiaoxia Champon and Ana-Maria Staicu and Anthony Weishampel and Chathura Jayalah and William Rand},
title = {Clustering Social Media Users Using Categorical-Valued Functional Data Analysis},
journal = {Journal of the American Statistical Association},
volume = {0},
number = {ja},
pages = {1--21},
year = {2026}
}

@article{tian:2024,
author = {Zhiyi Tian and Jiaming Xu and Jen Tang},
title = {Clustering High-Dimensional Noisy Categorical Data},
journal = {Journal of the American Statistical Association},
volume = {119},
number = {548},
pages = {3008--3019},
year = {2024},
publisher = {Taylor \& Francis},
doi = {10.1080/01621459.2023.2298028},
}

@article{argiento:2025,
author = {Raffaele Argiento and Edoardo Filippi-Mazzola and Lucia Paci},
title = {Model-Based Clustering of Categorical Data Based on the {H}amming Distance},
journal = {Journal of the American Statistical Association},
volume = {120},
number = {550},
pages = {1178--1188},
year = {2025},
publisher = {Taylor \& Francis},
doi = {10.1080/01621459.2024.2402568},
}

@article{ghilotti:2025,
    author = {Lorenzo Ghilotti and Mario Beraha and Alessandra Guglielmi},
    title = {Bayesian clustering of high-dimensional data via latent repulsive mixtures},
    journal = {Biometrika},
    volume = {112},
    number = {2},
    pages = {asae059},
    year = {2025},
    month = {04}
}

@article{rao:2025,
author = {Jackie Rao and Paul D.W. Kirk},
title = {{VICatMix}: variational {B}ayesian clustering and variable selection for discrete biomedical data},
journal = {Bioinformatics Advances},
volume = {5},
number = {1},
pages = {vbaf055},
year = {2025},
month = {01},
issn = {2635-0041},
doi = {10.1093/bioadv/vbaf055}
}

@article{papastamoulis:2017,
  author = {Panagiotis Papastamoulis and Magnus Rattray},
  title = {{BayesBinMix}: an {R} Package for Model Based Clustering of Multivariate Binary Data},
  journal = {The R Journal},
  year = {2017},
  note = {https://doi.org/10.32614/RJ-2017-022},
  doi = {10.32614/RJ-2017-022},
  volume = {9},
  issue = {1},
  issn = {2073-4859},
  pages = {403-420}
}

@article{bouguila:2009,
title = {Discrete data clustering using finite mixture models},
journal = {Pattern Recognition},
volume = {42},
number = {1},
pages = {33-42},
year = {2009},
issn = {0031-3203},
author = {Nizar Bouguila and Walid ElGuebaly},
}

@article{enjalbert:2025,
year = {2025},
author = {Nicolas Enjalbert Courrech and Cathy Maugis-Rabusseau and Pierre Neuvial},
title = {Review of Post-Clustering Inference Methods},
journal = {International Statistical Review}
}

@Article{papastamoulis:2023,
author={Panagiotis Papastamoulis },
title={Model based clustering of multinomial count data},
journal={Advances in Data Analysis and Classification},
year={2023},
month={Jul},
day={05},
}

@article{tang:2015,
title = {Model based clustering of high-dimensional binary data},
journal = {Computational Statistics \& Data Analysis},
volume = {87},
pages = {84-101},
year = {2015},
issn = {0167-9473},
doi = {https://doi.org/10.1016/j.csda.2014.12.009},
url = {},
author = {Yang Tang and Ryan P. Browne and Paul D. McNicholas},
}

@article{brini:2025,
author = {Alberto Brini and Abu Manju and Edwin R. van den Heuvel},
title = {A variable clustering approach for overdispersed high-dimensional count data using a copula-based mixture model},
journal = {Communications in Statistics - Simulation and Computation},
volume = {54},
number = {7},
pages = {2564--2584},
year = {2025},
publisher = {Taylor \& Francis},
doi = {10.1080/03610918.2024.2314666},
}

@Article{failli:2025,
author={Dalila Failli and Maria Francesca Marino
and Bruno Arpino},
title={Hierarchical Mixtures of Latent Trait Analyzers with concomitant variables for multivariate binary data},
journal={Statistics and Computing},
year={2025},
month={Aug},
day={28},
volume={35},
number={6},
pages={177},
}

@article{gormley:2023,
   author = {Isobel Claire Gormley  and T. Brendan Murphy and Adrian E. Raftery},
   title = {Model-Based Clustering}, 
   journal= {Annual Review of Statistics and Its Application},
   year = {2023},
   volume = {10},
   number = {Volume 10},
   pages = {573-595},
   doi = {https://doi.org/10.1146/annurev-statistics-033121-115326},
   url = {},
   publisher = {Annual Reviews},
   issn = {2326-831X},
   type = {Journal Article}
  }

@article{wade:2023,
    author = {Sara Wade},
    title = {Bayesian cluster analysis},
    journal = {Philosophical Transactions of the Royal Society A: Mathematical, Physical and Engineering Sciences},
    volume = {381},
    number = {2247},
    pages = {20220149},
    year = {2023},
    month = {03},
    issn = {1364-503X},
    doi = {10.1098/rsta.2022.0149},
    url = {},
    eprint = {https://royalsocietypublishing.org/rsta/article-pdf/doi/10.1098/rsta.2022.0149/1327776/rsta.2022.0149.pdf},
}

@article{chandra:2023,
  author  = {Noirrit Kiran Chandra and Antonio Canale and David B. Dunson},
  title   = {Escaping The Curse of Dimensionality in {Bayesian} Model-Based Clustering},
  journal = {Journal of Machine Learning Research},
  year    = {2023},
  volume  = {24},
  number  = {144},
  pages   = {1--42},
  url     = {}
}

@article{chen:2023,
  author  = {Yiqun T. Chen and Daniela M. Witten},
  title   = {Selective inference for k-means clustering},
  journal = {Journal of Machine Learning Research},
  year    = {2023},
  volume  = {24},
  number  = {152},
  pages   = {1--41},
  url     = {}
}

@article{krazanowski:2009,
title = {A simple method for screening variables before clustering microarray data},
journal = {Computational Statistics \& Data Analysis},
volume = {53},
number = {7},
pages = {2747-2753},
year = {2009},
issn = {0167-9473},
doi = {https://doi.org/10.1016/j.csda.2009.02.001},
url = {},
author = {Wojtek J. Krzanowski and David J. Hand}
}

@article{kiselev:2019,
  author  = { Vladimir Yu Kiselev and Tallulah S. Andrews and Martin Hemberg},
  title   = {Challenges in unsupervised clustering of single-cell {RNA}-seq data},
  journal = {Nature Reviews Genetics},
  volume  = {20},
  number  = {5},
  pages   = {273--282},
  year    = {2019},
  doi     = {10.1038/s41576-018-0088-9}
}

@article{grabski:2023,
  author  = {Isabella N. Grabski and Kelly Street and Rafael A. Irizarry},
  title   = {Significance analysis for clustering with single-cell {RNA}-sequencing data},
  journal = {Nature Methods},
  volume  = {20},
  number  = {8},
  pages   = {1196--1202},
  year    = {2023},
  doi     = {10.1038/s41592-023-01933-9}
}

@article{weller:2020,
  author  = {Bridget E. Weller and Natasha K. Bowen  and Sarah J. Faubert},
  title   = {Latent class analysis: A guide to best practice},
  journal = {Journal of Black Psychology},
  volume  = {46},
  number  = {4},
  pages   = {287--311},
  year    = {2020},
  doi     = {10.1177/0095798420930932}
}

@article{dinh:2025,
title = {Categorical data clustering: 25 years beyond K-modes},
journal = {Expert Systems with Applications},
volume = {272},
pages = {126608},
year = {2025},
issn = {0957-4174},
doi = {https://doi.org/10.1016/j.eswa.2025.126608},
url = {},
author = {Tai Dinh and Hauchi Wong and Philippe Fournier-Viger and Daniil Lisik and Minh-Quyet Ha and Hieu-Chi Dam and Van-Nam Huynh},
}

@article{lyu:2026, title={Spectral Clustering with Likelihood Refinement for High-Dimensional Latent Class Recovery}, volume={91}, DOI={10.1017/psy.2026.10095}, number={2}, journal={Psychometrika}, author={Zhongyuan Lyu and Yuqi Gu}, year={2026}, pages={695–723}}

@Article{huang:1998,
author={Zhexue Huang},
title={Extensions to the k-Means Algorithm for Clustering Large Data Sets with Categorical Values},
journal={Data Mining and Knowledge Discovery},
year={1998},
month={Sep},
day={01},
volume={2},
number={3},
pages={283-304},
}

@article{guha:2000,
title = {Rock: A robust clustering algorithm for categorical attributes},
journal = {Information Systems},
volume = {25},
number = {5},
pages = {345-366},
year = {2000},
issn = {0306-4379},
doi = {https://doi.org/10.1016/S0306-4379(00)00022-3},
url = {},
author = {Sudipto Guha and Rajeev Rastogi and Kyuseok Shim},
}

@article{fop:2018,
  title={Variable selection methods for model-based clustering},
  author={Michael Fop and T. Brendan Murphy},
  journal={Statistics Surveys},
  volume={12},
  pages={18--65},
  year={2018},
  publisher={Amer. Statist. Assoc., the Bernoulli Soc., the Inst. Math. Statist., and the~…}
}

@article{babu:2025,
	author = {Ganesh Babu and Aoife Gowen and Michael Fop and Isobel Claire Gormley},
	date = {2025/06/01},
	doi = {10.1007/s11634-025-00623-y},
	id = {Babu2025},
	isbn = {1862-5355},
	journal = {Advances in Data Analysis and Classification},
	number = {2},
	pages = {323--359},
	title = {A consensus-constrained parsimonious Gaussian mixture model for clustering hyperspectral images},
	volume = {19},
	year = {2025},
}

@article{mori:2026,
	author = {Matteo Mori and Laura Anderlucci},
	date = {2026/05/21},
	doi = {10.1007/s11634-026-00684-7},
	id = {Mori2026},
	isbn = {1862-5355},
	journal = {Advances in Data Analysis and Classification},
	title = {Model-based clustering of functional data via random projection ensembles},
	year = {2026}}

@Manual{unvotes,
    title = {unvotes: United Nations General Assembly Voting Data},
    author = {David Robinson},
    year = {2021},
    note = {R package version 0.3.0},
    url = {https://CRAN.R-project.org/package=unvotes},
    doi = {10.32614/CRAN.package.unvotes},
  }

@incollection{voeten:2013,
  author    = {Erik Voeten},
  title     = {Data and Analyses of Voting in the {UN} General Assembly},
  booktitle = {Routledge Handbook of International Organization},
  editor    = {Reinalda, Bob},
  publisher = {Routledge},
  address   = {London},
  year      = {2013},
  pages     = {54--66}
}

@article{chang:1983,
  author  = { Wei-Chien Chang},
  title   = {On Using Principal Components before Separating a Mixture of Two Multivariate Normal Distributions},
  journal = {Journal of the Royal Statistical Society: Series C (Applied Statistics)},
  year    = {1983},
  volume  = {32},
  number  = {3},
  pages   = {267--275},
  doi     = {10.2307/2347949}
}

@article{hall:2005,
  author  = {Peter  Hall and J. S. Marron and Amnon Neeman},
  title   = {Geometric Representation of High Dimension, Low Sample Size Data},
  journal = {Journal of the Royal Statistical Society: Series B (Statistical Methodology)},
  year    = {2005},
  volume  = {67},
  number  = {3},
  pages   = {427--444},
  doi     = {10.1111/j.1467-9868.2005.00510.x}
}

@article{johnstone:2009,
  author  = {Iain M. Johnstone and Arthur Yu Lu},
  title   = {On Consistency and Sparsity for Principal Components Analysis in High Dimensions},
  journal = {Journal of the American Statistical Association},
  year    = {2009},
  volume  = {104},
  number  = {486},
  pages   = {682--693},
  doi     = {10.1198/jasa.2009.0121}
}

@article{fan:2014,
  author  = {Jianqing Fan and Fang Han and Han Liu},
  title   = {Challenges of {Big Data} Analysis},
  journal = {National Science Review},
  year    = {2014},
  volume  = {1},
  number  = {2},
  pages   = {293--314},
  doi     = {10.1093/nsr/nwt032}
}

@article{cilibrasi:2005,
  author  = {Rudi Cilibrasi and Paul M.B. Vit{\'a}ny},
  title   = {Clustering by Compression},
  journal = {IEEE Transactions on Information Theory},
  year    = {2005},
  volume  = {51},
  number  = {4},
  pages   = {1523--1545},
  doi     = {10.1109/TIT.2005.844059}
}

@article{ma:2007,
  author  = {Yi Ma and Harm Derksen and Wei Hong and John Wright},
  title   = {Segmentation of Multivariate Mixed Data via Lossy Data Coding and Compression},
  journal = {IEEE Transactions on Pattern Analysis and Machine Intelligence},
  year    = {2007},
  volume  = {29},
  number  = {9},
  pages   = {1546--1562},
  doi     = {10.1109/TPAMI.2007.1085}
}

@article{feldman:2020,
  author  = {Dan Feldman and Melanie Schmidt and Christin Sohler},
  title   = {Turning Big Data into Tiny Data: Constant-Size Coresets for k-Means, {PCA}, and Projective Clustering},
  journal = {SIAM Journal on Computing},
  year    = {2020},
  volume  = {49},
  number  = {3},
  pages   = {601--657},
  doi     = {10.1137/18M1209854}
}

@article{keriven:2018,
  author  = {Nicolas Keriven and Anthony Bourrier and R{\'e}mi Gribonval and Patrick P{\'e}rez},
  title   = {Sketching for Large-Scale Learning of Mixture Models},
  journal = {Information and Inference: A Journal of the IMA},
  year    = {2018},
  volume  = {7},
  number  = {3},
  pages   = {447--508},
  doi     = {10.1093/imaiai/iax015}
}

@incollection{johnson:1984,
  author    = {William B. Johnson and Joram Lindenstrauss},
  title     = {Extensions of {L}ipschitz Mappings into a {H}ilbert Space},
  booktitle = {Conference in Modern Analysis and Probability (New Haven, Conn., 1982)},
  series    = {Contemporary Mathematics},
  volume    = {26},
  pages     = {189--206},
  publisher = {American Mathematical Society},
  address   = {Providence, RI},
  year      = {1984},
  doi       = {10.1090/conm/026/737400}
}

@article{boutsidis:2015,
  author  = {Christos Boutsidis and Anastasios Zouzias and Michael W. Mahoney and Petros Drineas},
  title   = {Randomized Dimensionality Reduction for $k$-Means Clustering},
  journal = {IEEE Transactions on Information Theory},
  year    = {2015},
  volume  = {61},
  number  = {2},
  pages   = {1045--1062},
  doi     = {10.1109/TIT.2014.2375327}
}

@inproceedings{weinberger:2009,
  author    = {Kilian Q. Weinberger and Anirban Dasgupta and John Langford and Alexander J. Smola and Josh Attenberg},
  title     = {Feature Hashing for Large Scale Multitask Learning},
  booktitle = {Proceedings of the 26th Annual International Conference on Machine Learning},
  year      = {2009},
  pages     = {1113--1120},
  publisher = {ACM},
  doi       = {10.1145/1553374.1553516}
}

@article{charikar:2004,
  author  = {Moses Charikar and Kevin Chen and Martin Farach-Colton},
  title   = {Finding Frequent Items in Data Streams},
  journal = {Theoretical Computer Science},
  year    = {2004},
  volume  = {312},
  number  = {1},
  pages   = {3--15},
  doi     = {10.1016/S0304-3975(03)00400-6}
}

@article{clarkson:2017,
  author  = {Kenneth L. Clarkson and David P. Woodruff},
  title   = {Low-Rank Approximation and Regression in Input Sparsity Time},
  journal = {Journal of the ACM},
  year    = {2017},
  volume  = {63},
  number  = {6},
  pages   = {54:1--54:45},
  doi     = {10.1145/3019134}
}

@inproceedings{pratap:2019,
  author    = {Rameshwar Pratap and Debajyoti Bera and Karthik Revanuru},
  title     = {Efficient Sketching Algorithm for Sparse Binary Data},
  booktitle = {Proceedings of the 2019 IEEE International Conference on Data Mining (ICDM)},
  year      = {2019},
  pages     = {508--517},
  publisher = {IEEE}
}

@article{verma:2022,
  author  = {Bhisham Dev Verma and Rameshwar Pratap and Debajyoti Bera},
  title   = {Efficient Binary Embedding of Categorical Data Using {BinSketch}},
  journal = {Data Mining and Knowledge Discovery},
  year    = {2022},
  volume  = {36},
  number  = {2},
  pages   = {537--565},
  doi     = {10.1007/s10618-021-00815-y}
}

@article{bera:2023,
  author  = {Debajyoti Bera and Rameshwar Pratap and Bhisham Dev Verma},
  title   = {Dimensionality Reduction for Categorical Data},
  journal = {IEEE Transactions on Knowledge and Data Engineering},
  year    = {2023},
  volume  = {35},
  number  = {4},
  pages   = {3658--3671}
}

\clearpage
\appendix

\section{Proofs}
\subsection{Proof of Proposition 1}
{\bf Proof}
Let $t=K^{1/(pK)}$ and define from Equation 2:
 \[
(z_i - z_h) =\sum_{j=1}^p t^{p-j} (y_{ij} -y_{hj} )= \sum_{j=1}^p t^{p-j} r_{ihj} , \quad i, h=1, \dots, n, i\neq h.
\]
This expression is a polynomial in $t$ of degree at most $p-1$, with integer coefficients $r_{ihj}$. Since $t$ is irrational and its minimal polynomial over $\mathbb{Q}$ has degree of at least $p$ \citep{lang:2002}, the polynomial can vanish only if all coefficients are zero, that is only only when $\mathbf{y}_i=\mathbf{y}_h$.

\subsection{Proof of Proposition 2}
{\bf Proof}
For block $s$, denoted by $B_s$, 
\[\left(z_{is}-z_{hs} \right)= \sum_{j \in B_s} K^{\frac{p_s-j}{p_sK}} \left( y_{ij} -y_{hj}\right) = \sum_{j \in B_s}t_j \left( y_{ij} -y_{hj}\right) =
\sum_{j \in B_s} t_j \Delta_{js} ,\]
with $1 \leq t_j < K^{\frac{1}{K}}$. By the triangle inequality:
\[
d_{ih}^* = \left( \sum_{s =1}^{q} |z_{is}-z_{hs}|^2\right)^{\frac{1}{2}}\leq \left( \sum_{s =1}^{q} \left(\sum_{j \in B_s} t_j |\Delta_{js}|  \right)^2\right)^{\frac{1}{2}} \leq
K^{\frac{1}{K}}\left( \sum_{s =1}^{q}\left(\sum_{j \in B_s}  |\Delta_{js}| \right)^2\right)^{\frac{1}{2}}
\]
Therefore, as $L^1$ distances are an upper bound for $L^2$ distances computed on the same vector, we have:
\[
d_{ih}^* \leq K^{\frac{1}{K} }  \sum_{s =1}^{q} \sum_{j \in B_s} |\Delta_{js}| = K^{\frac{1}{K} } d_{ih} \leq 3^{\frac{1}{3} } d_{ih}.
\]

\subsection{Proof of Proposition 3}

{\bf Proof}
The difference between any two compressed cluster centroids is:
\[
\begin{split}
   u_g - u_l &= \frac{1}{n_g} \sum_{i} c_{ig}z_{i} - \frac{1}{n_{l}} \sum_{i}c_{il} z_{i} \\
    &=
    \frac{1}{n_g} \sum_{i}c_{ig} \sum_{j=1}^p K^{\frac{p-j}{Kp}} y_{ij}  - \frac{1}{n_l} \sum_{i}c_{il} \sum_{j=1}^p K^{\frac{p-j}{Kp}} y_{ij} \\
    &= \sum_{j=1}^p K^{\frac{p-j}{Kp}} \left( \frac{1}{n_g} \sum_{i}c_{ig} y_{ij} - \frac{1}{n_{l}} \sum_{i}c_{il} y_{ij} \right)\\
    &=\sum_{j=1}^p K^{\frac{p-j}{Kp}} \left( m_{gj}  -  m_{lj} \right).
\end{split}
\]
The expression above is a polynomial in powers of $K^{1/(pK)}$, with rational coefficients given by the empirical cluster centroid differences $m_{gj}-m_{lj}$. By the same argument used in the proof of Proposition 2.1, this polynomial can vanish only if all coefficients are zero. Hence, if $\mathbf{m}_g$ and $\mathbf{m}_l$ differ in at least one coordinate, then $u_g \neq u_l$.

\section{Simulation study}

The following figures provide the full results for the simulation study presented in Section 3. They report computational times, clustering performance measured by the adjusted Rand index (ARI), and representative low-dimensional projections obtained with the different methods, i.e. data compression (DC), principal component analysis (PCA), multidimensional scaling (MDS), and t-distributed stochastic neighbor embedding (t-SNE).

\begin{figure}
\centering
    \includegraphics[scale=.56]{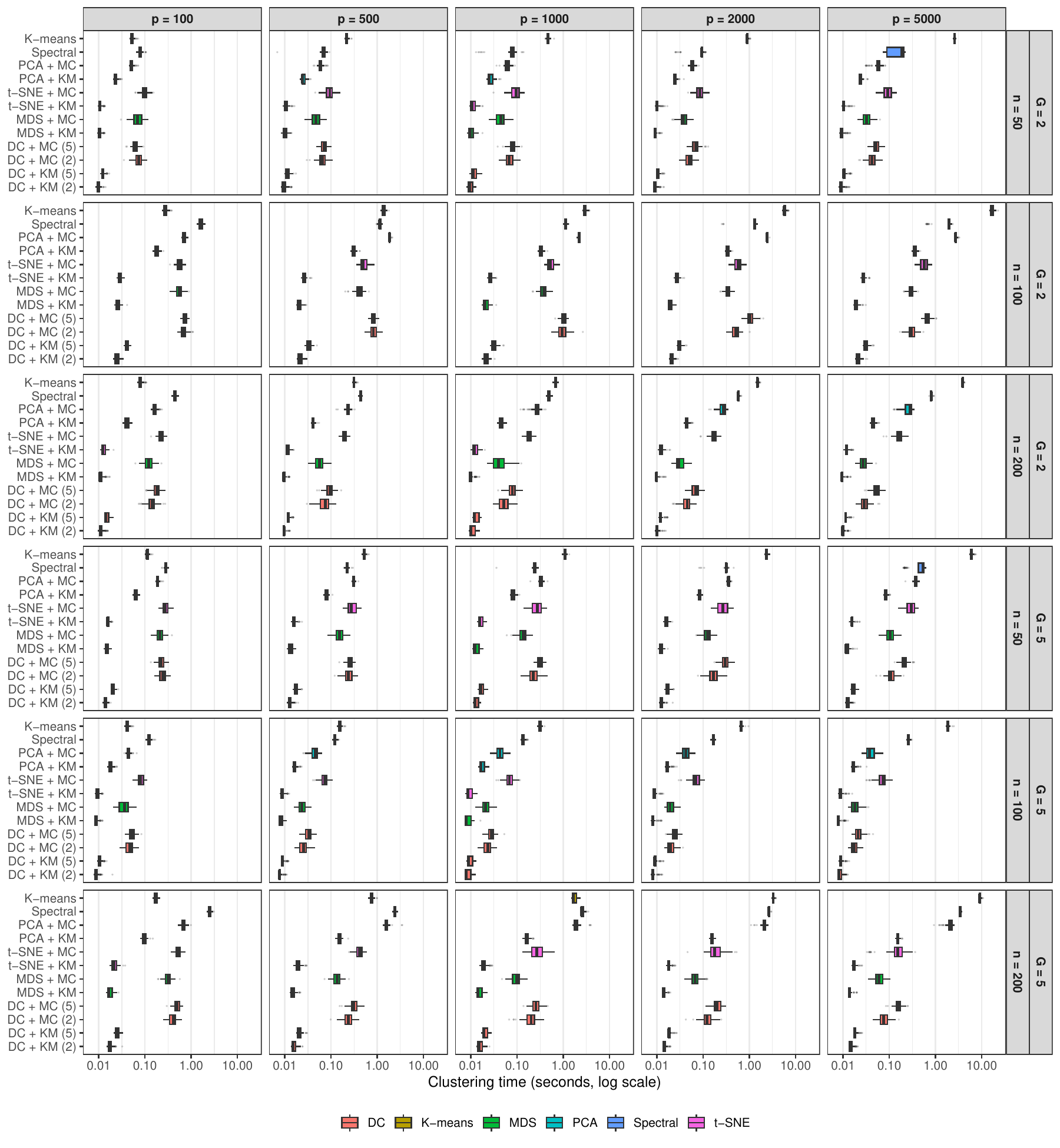}
    \caption{\label{fig:sim1_times} Simulation study -- Scenario 1. Computational times for the different methods.}
\end{figure}

\begin{figure}
\centering
    \includegraphics[scale=.8]{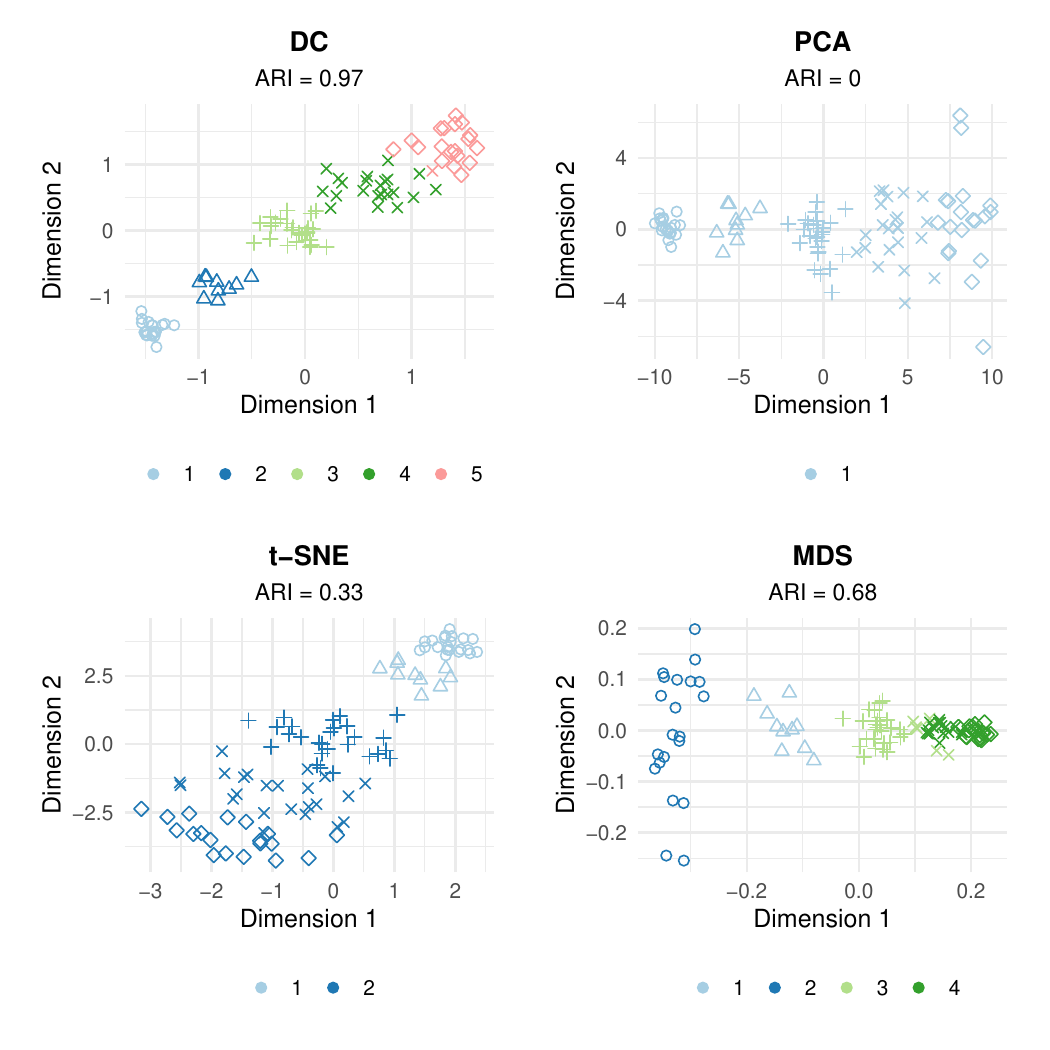}
    \caption{\label{fig:sim1_example} Simulation study -- Scenario 1. Example of low-dimensional representation for $G=5$, $n=100$, and $p=100$, obtained using DC, PCA, t-SNE, and MDS. Symbols indicate simulated true classes, while colors indicate clusters estimated by \texttt{mclust}. Adjusted Rand index (ARI) values are also reported.}
\end{figure}

\begin{figure}
\centering
    \includegraphics[scale=.6]{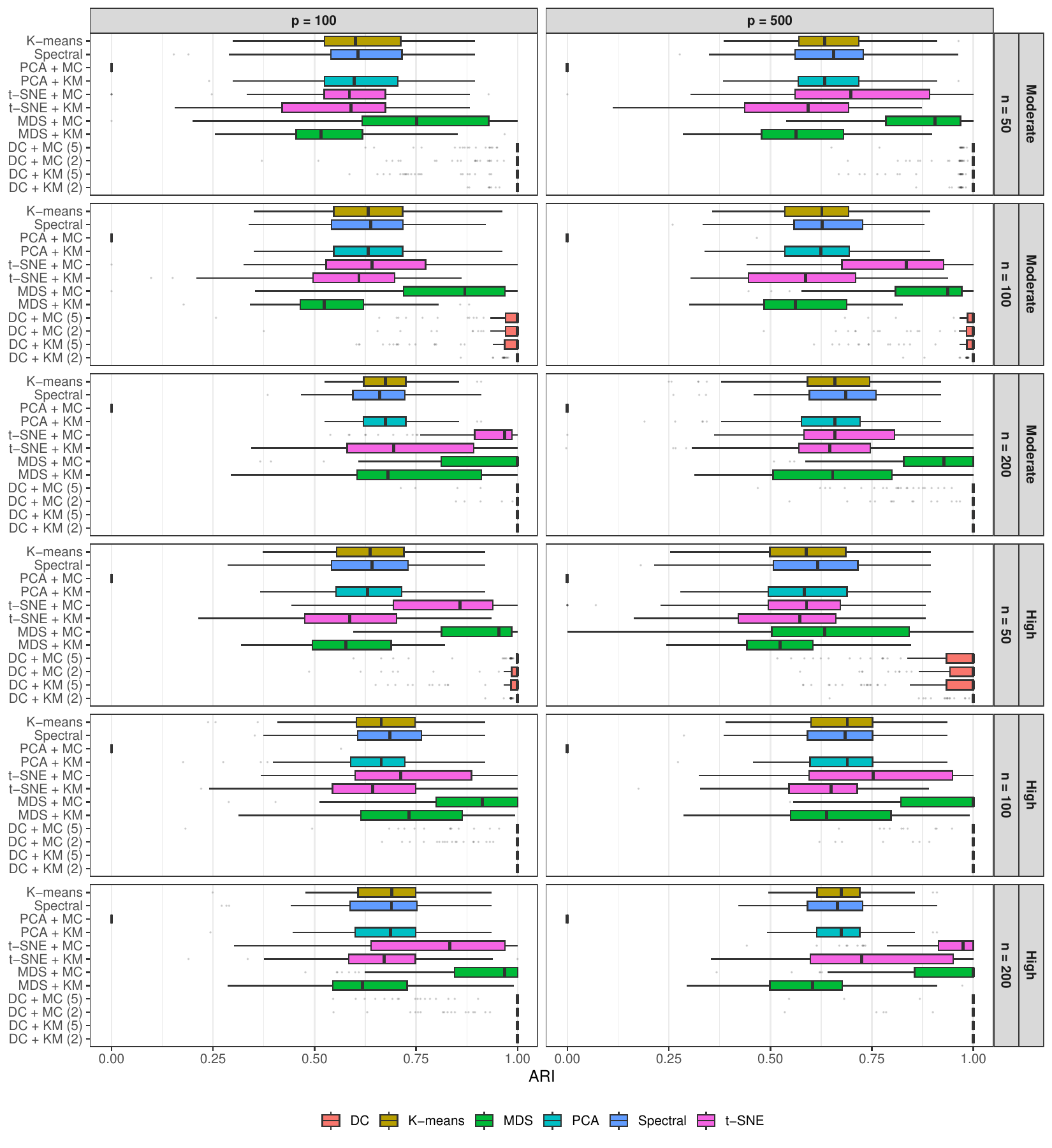}
    \caption{\label{fig:sim2_appendix} Simulation study -- Scenario 2. ARI values between simulated true and estimated class memberships for the different methods.}
\end{figure}

\begin{figure}
\centering
    \includegraphics[scale=.8]{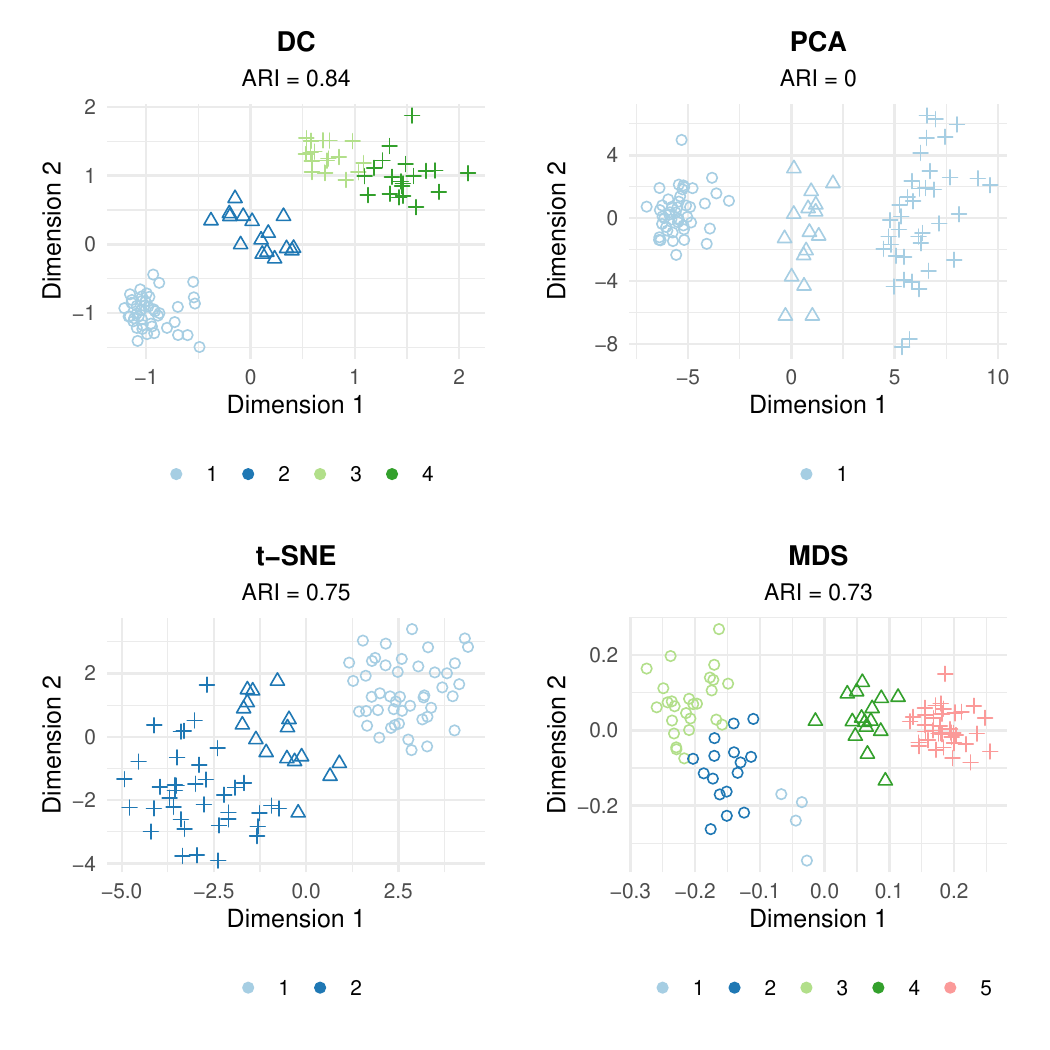}
    \caption{\label{fig:sim2_example} Simulation study -- Scenario 2. Example of low-dimensional representation for $n=100$ and $p=100$ in the high-correlation setting, obtained using DC, PCA, t-SNE, and MDS. Symbols indicate simulated true classes, while colors indicate clusters estimated by \texttt{mclust}. Adjusted Rand index (ARI) values are also reported.}
\end{figure}

\begin{figure}
\centering
    \includegraphics[scale=.6]{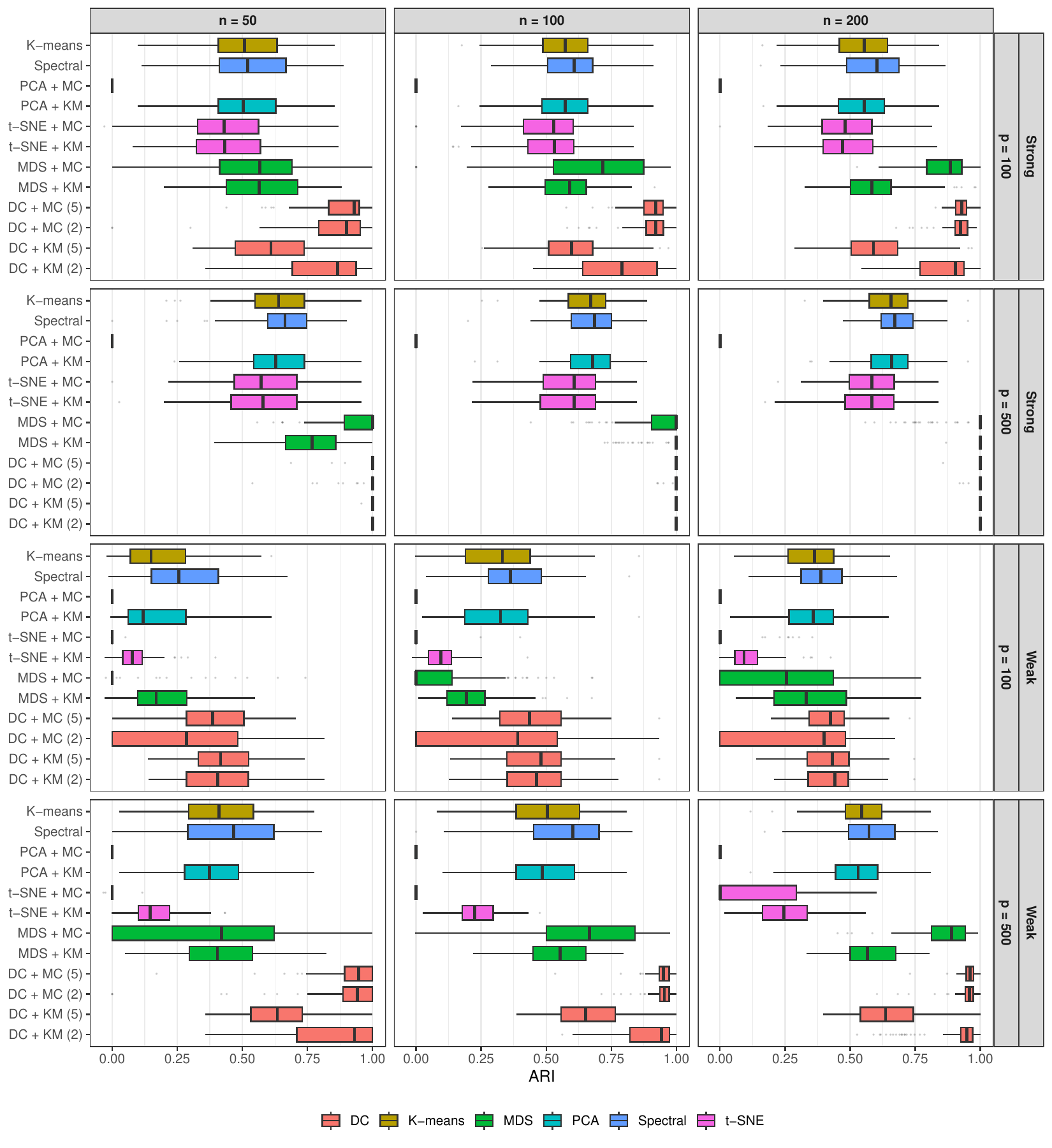}
    \caption{\label{fig:sim3_appendix} Simulation study -- Scenario 3. ARI values between simulated true and estimated class memberships for the different methods.}
\end{figure}

\begin{figure}
\centering
    \includegraphics[scale=.8]{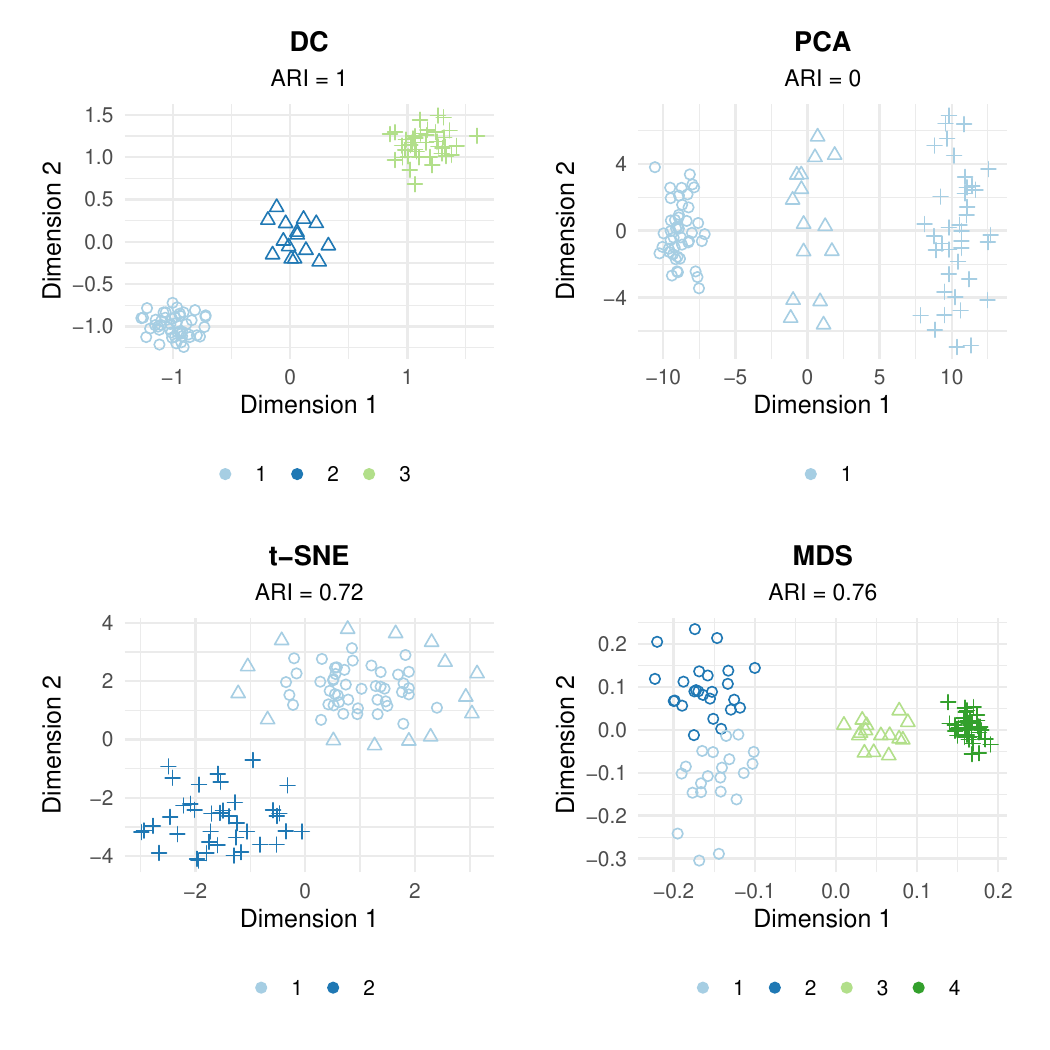}
    \caption{\label{fig:sim3_example} Simulation study -- Scenario 3. Example of low-dimensional representation for $n=100$ and $p=500$ in the high-separation setting, obtained using DC, PCA, t-SNE, and MDS. Symbols indicate simulated true classes, while colors indicate clusters estimated by \texttt{mclust}. Adjusted Rand index (ARI) values are also reported.}
\end{figure}

\begin{figure}
\centering
    \includegraphics[scale=.6]{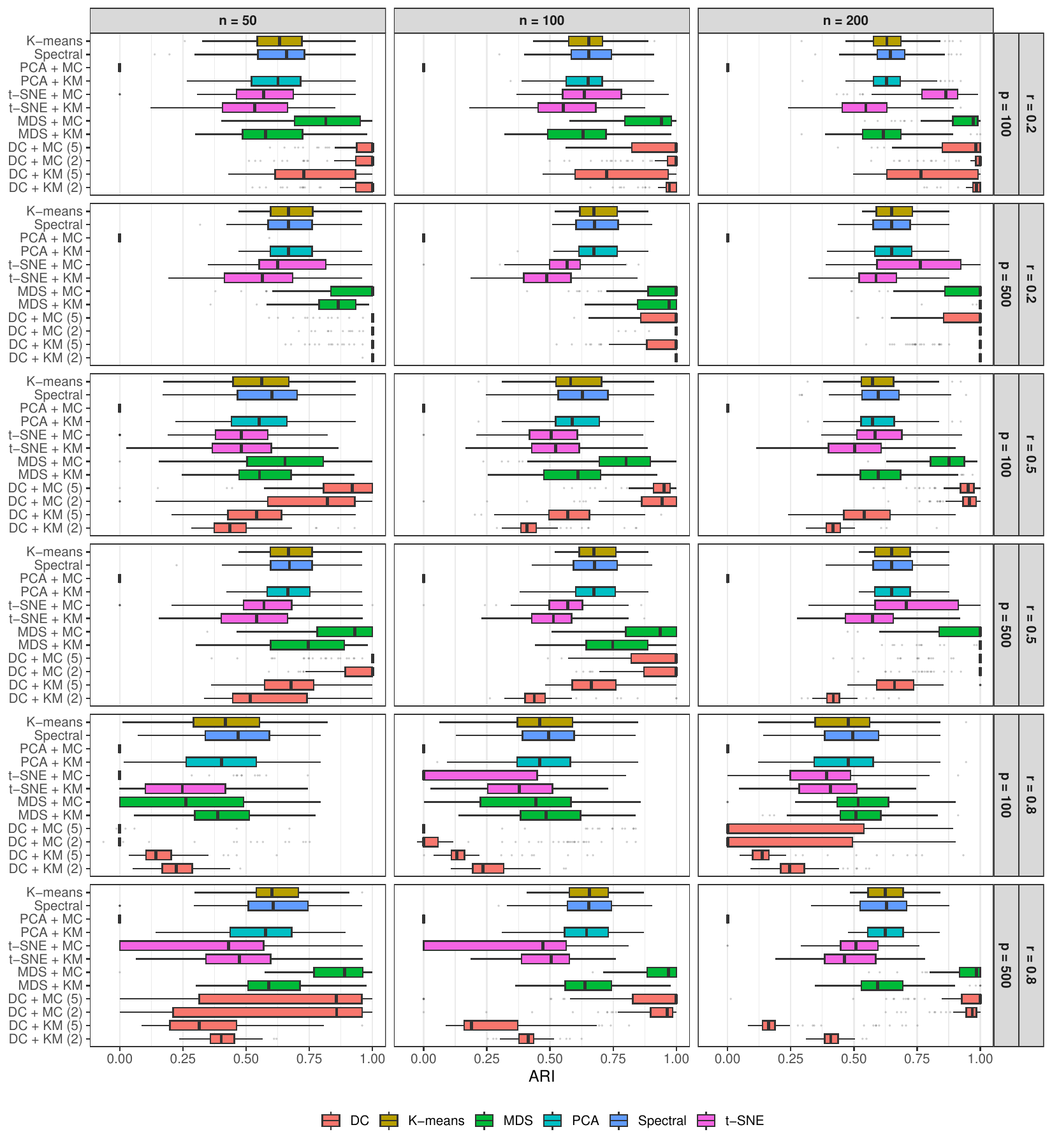}
    \caption{\label{fig:sim41_appendix} Simulation study -- Scenario 4. ARI values between simulated true and estimated class memberships for the different methods, for $p \in \{100,500\}$.}
\end{figure}

\begin{figure}
\centering
    \includegraphics[scale=.6]{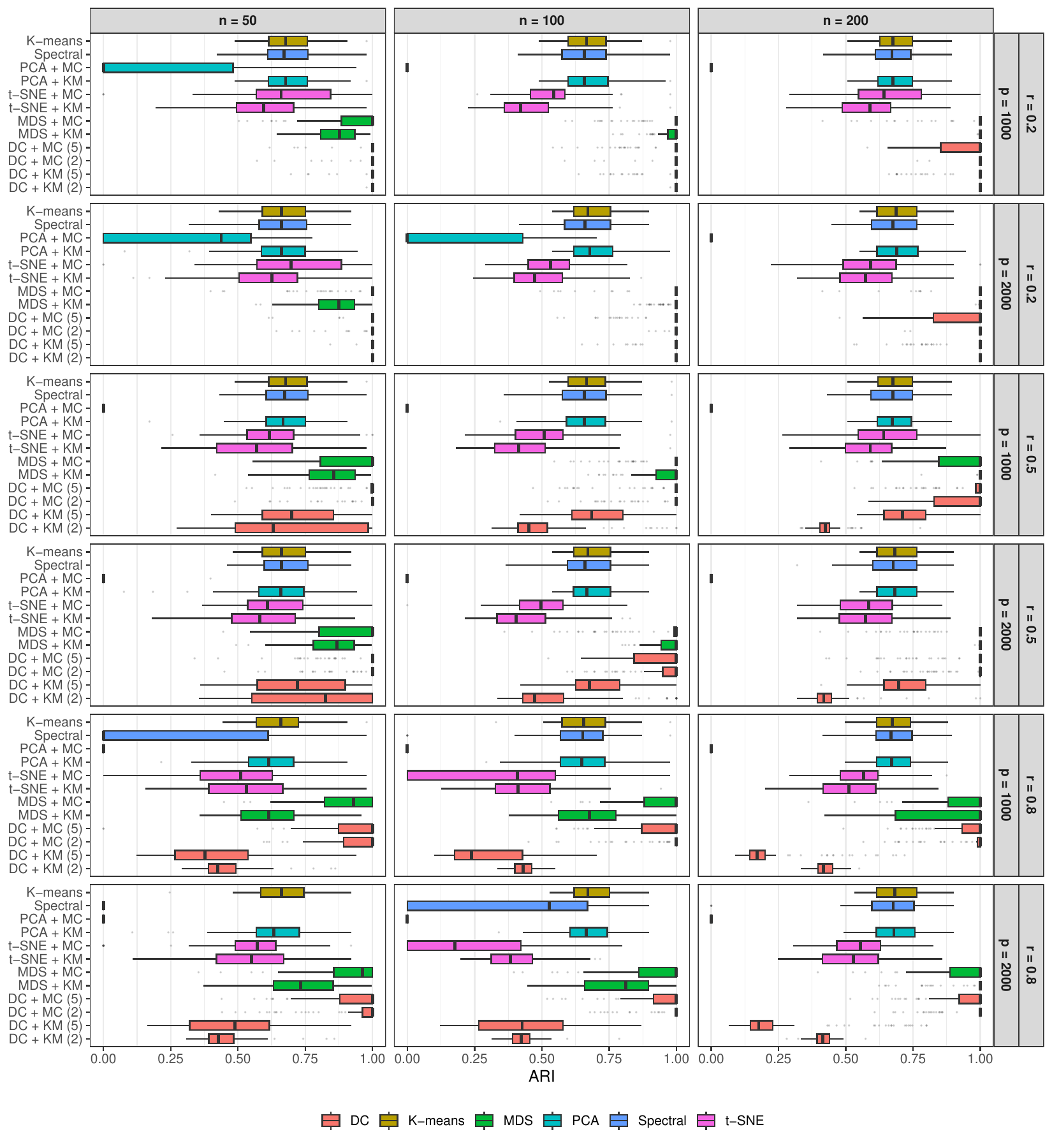}
    \caption{\label{fig:sim42_appendix} Simulation study -- Scenario 4. ARI values between simulated true and estimated class memberships for the different methods, for $p \in \{1000,2000\}$.}
\end{figure}

\begin{figure}
\centering
    \includegraphics[scale=.8]{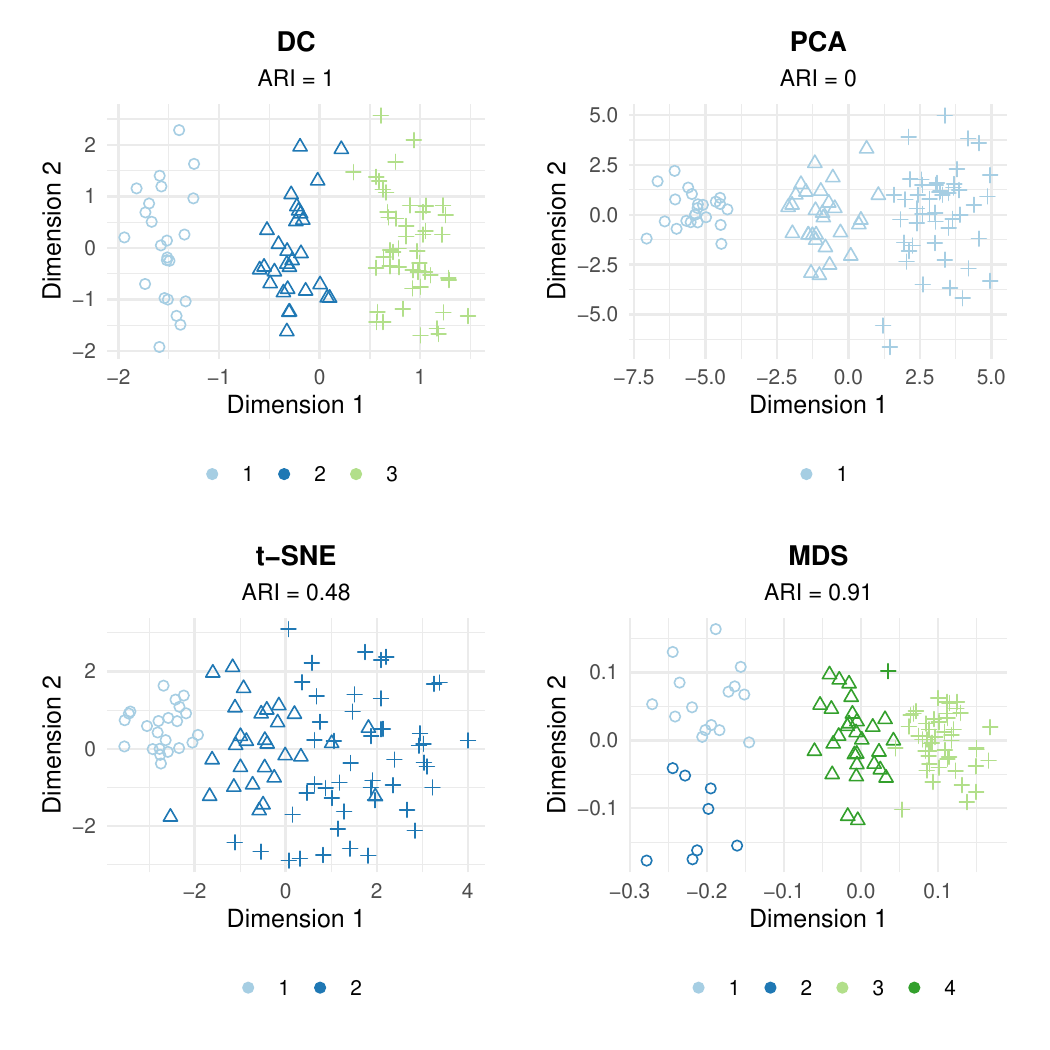}
    \caption{\label{fig:sim4_example} Simulation study -- Scenario 4. Example of low-dimensional representation for $n=100$, $p=100$, and noise proportion $r=0.5$, obtained using DC, PCA, t-SNE, and MDS. Symbols indicate simulated true classes, while colors indicate clusters estimated by \texttt{mclust}. Adjusted Rand index (ARI) values are also reported.}
\end{figure}

\begin{figure}
\centering
    \includegraphics[scale=.6]{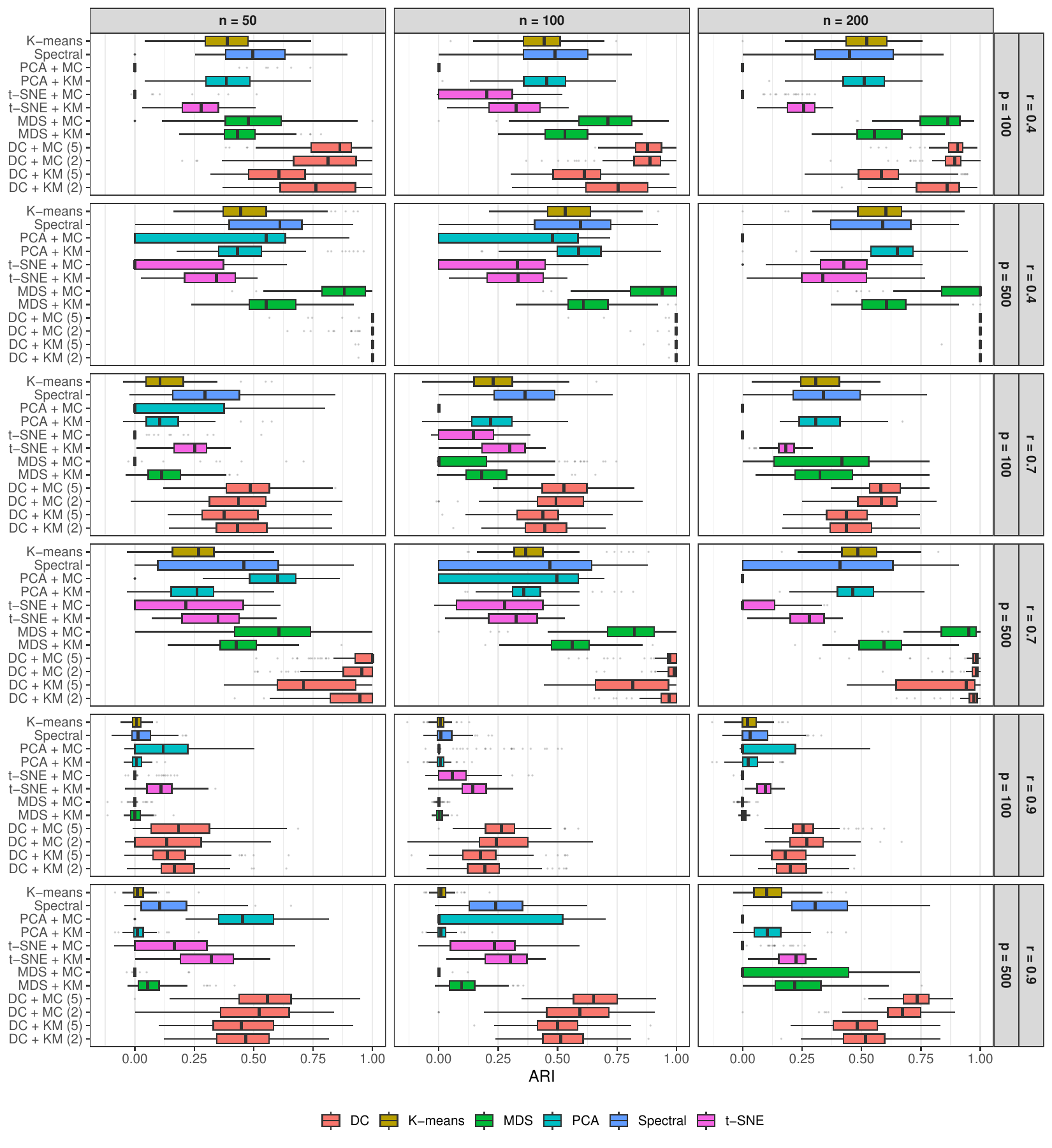}
    \caption{\label{fig:sim51_appendix} Simulation study -- Scenario 5. ARI values between simulated true and estimated class memberships for the different methods, for $p \in \{100,500\}$.}
\end{figure}

\begin{figure}
\centering
    \includegraphics[scale=.6]{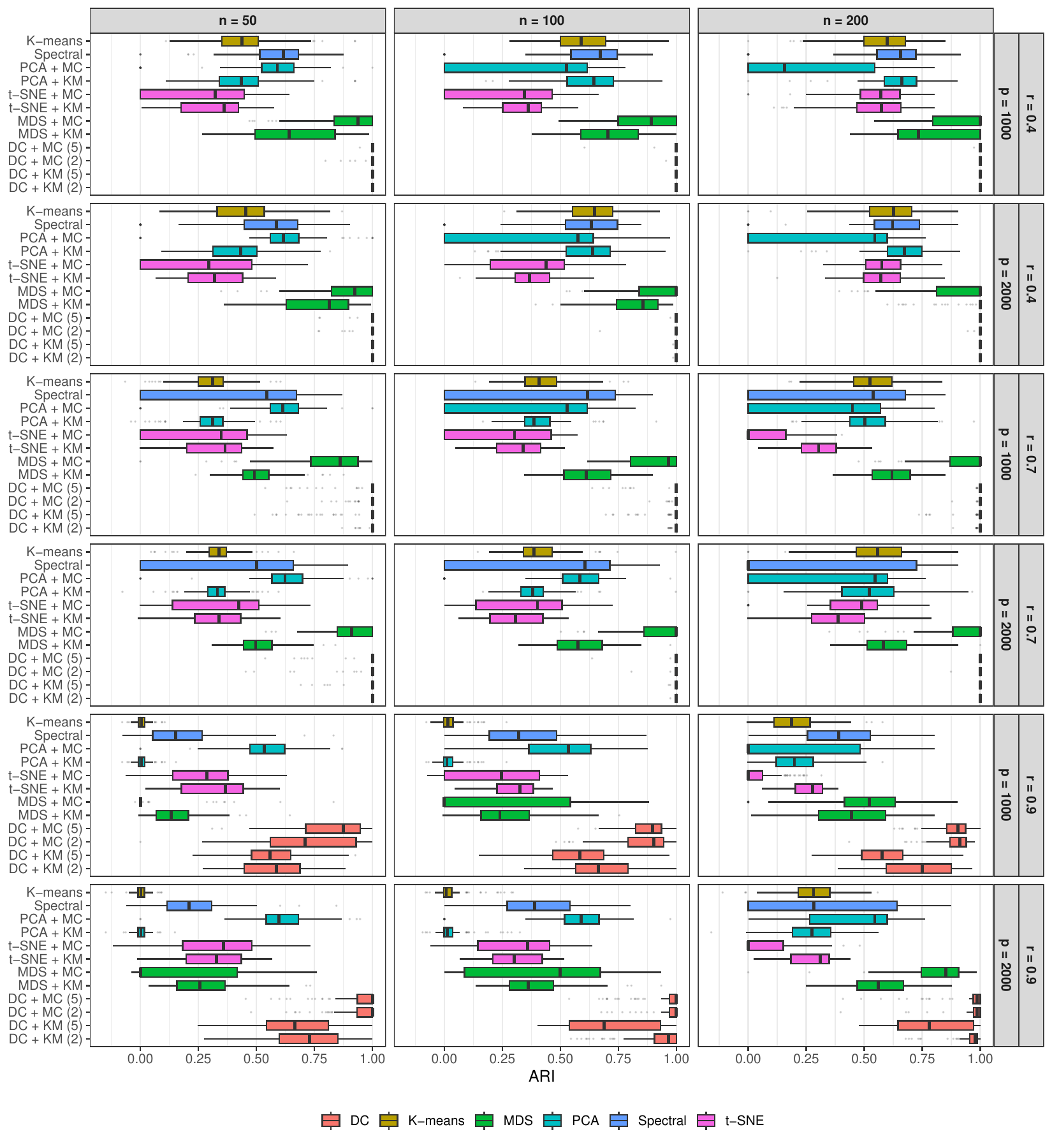}
    \caption{\label{fig:sim52_appendix} Simulation study -- Scenario 5. ARI values between simulated true and estimated class memberships for the different methods, for $p \in \{1000,2000\}$.}
\end{figure}

\begin{figure}
\centering
    \includegraphics[scale=.8]{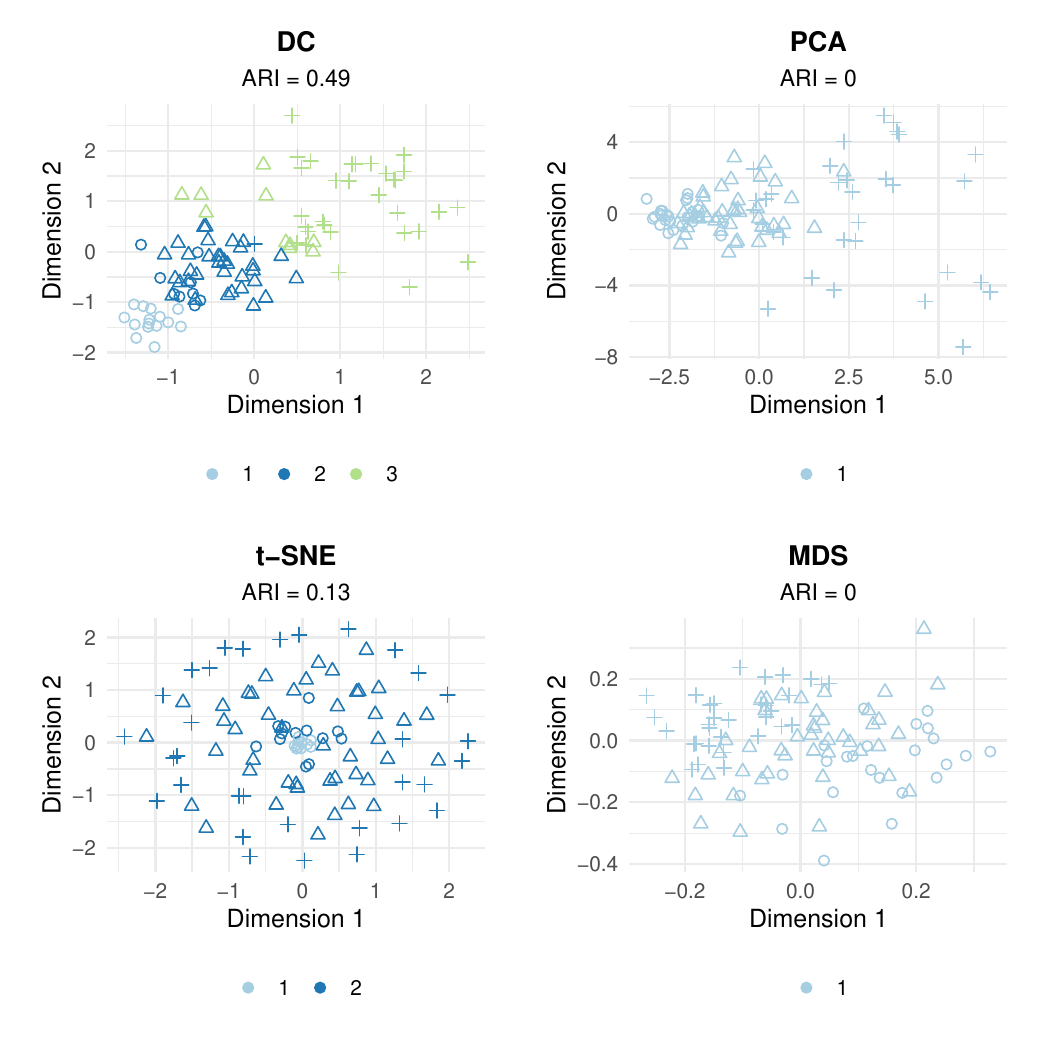}
    \caption{\label{fig:sim5_example} Simulation study -- Scenario 5. Example of low-dimensional representation for $n=100$, $p=100$, and sparsity level $r=0.7$, obtained using DC, PCA, t-SNE, and MDS. Symbols indicate simulated true classes, while colors indicate clusters estimated by \texttt{mclust}. Adjusted Rand index (ARI) values are also reported.}
\end{figure}

\section{Illustrative applications}

\subsection{United Nations voting data}

\subsubsection{DC-based solution}
Below, in alphabetical order, the list of countries contained in the seven clusters retrieved using DC-compressed data.
\begin{itemize}
    \item[Cluster 1: ] Albania, Cambodia, Central African Republic, Chad, Comoros, Congo - Kinshasa, Czechoslovakia, Dominica, Equatorial Guinea, Eswatini, Federal Republic of Germany, Gambia, German Democratic Republic, Guatemala, Haiti, Iraq, Israel, Kiribati, Liberia, Malawi, Nauru, Rwanda, São Tomé and  Príncipe, Seychelles, Somalia, South Africa, South Sudan, St. Kitts and Nevis, Taiwan, United States, Yemen Arab Republic, Yemen People's Republic, Yugoslavia, Zanzibar.
    \item[Cluster 2: ] Australia, Austria, Belgium, Canada, Denmark, Finland, France, Iceland, Ireland, Italy, Japan, Luxembourg, Netherlands, New Zealand, Norway, Paraguay, Portugal, Spain, Sweden, United Kingdom.
 \item[Cluster 3: ]  Belarus, Bolivia, Brazil, Colombia, Cuba, Ecuador, Egypt, Ethiopia, Guinea, Honduras, India, Indonesia, Iran, Lebanon, Mexico, Nicaragua, Pakistan, Panama, Peru, Philippines, Poland, Russia, Saudi Arabia, Syria, Thailand, Ukraine, Uruguay, Venezuela.
\item[Cluster 4: ] Afghanistan, Algeria, Argentina, Bahrain, Barbados, Benin, Bhutan, Botswana, Bulgaria, Burkina Faso, Burundi, Cameroon, Chile, China, Congo - Brazzaville, Costa Rica, Côte d’Ivoire, Cyprus, Dominican Republic, El Salvador, Fiji, Gabon, Ghana, Greece, Guinea-Bissau, Guyana, Hungary, Jamaica, Jordan, Kenya, Kuwait, Laos, Lesotho, Libya, Madagascar, Malaysia, Maldives, Mali, Malta, Mauritania, Mauritius, Mongolia, Morocco, Myanmar (Burma), Nepal, Niger, Nigeria, Oman, Qatar, Romania, Senegal, Sierra Leone, Singapore, Sri Lanka, Sudan, Tanzania, Togo, Trinidad and Tobago, Tunisia, Turkey, Uganda, United Arab Emirates, Zambia.
\item[Cluster 5: ] Azerbaijan, Bosnia and Herzegovina, Eritrea, Kazakhstan, Kyrgyzstan, Liechtenstein, Marshall Islands, Micronesia (Federated States of), Montenegro, Namibia, North Korea, Palau, San Marino, Switzerland, Tajikistan, Timor-Leste, Tonga, Turkmenistan, Tuvalu, Uzbekistan, Yemen.
\item[Cluster 6: ] Angola, Antigua and Barbuda, Bahamas, Bangladesh, Belize, Brunei, Cape Verde, Djibouti, Grenada, Mozambique, Papua New Guinea, Samoa, Solomon Islands, St. Lucia, St. Vincent and Grenadines, Suriname, Vanuatu, Vietnam, Zimbabwe.
\item[Cluster 7: ] Andorra, Armenia, Croatia, Czechia, Estonia, Georgia, Germany, Latvia, Lithuania, Moldova, Monaco, North Macedonia, Slovakia, Slovenia, South Korea.
\end{itemize}

\subsubsection{MDS-based solution}
Below, in alphabetical order, the list of countries contained in the seven clusters retrieved using the MDS representation.
\begin{itemize}
    \item[Cluster 1: ] Albania, Belarus, Bulgaria, Cambodia, Cameroon, Central African Republic, China, Congo - Kinshasa, Cyprus, Czechoslovakia, Equatorial Guinea, Federal Republic of Germany, Gambia, German Democratic Republic, Hungary, Indonesia, Iraq, Israel, Liberia, Malta, Mongolia, Poland, Romania, Russia, Rwanda, Somalia, South Africa, Taiwan, Ukraine, United States, Yemen Arab Republic, Yemen People's Republic, Yugoslavia.
    \item[Cluster 2: ] Australia, Austria, Belgium, Canada, Denmark, Finland, France, Greece, Iceland, Ireland, Italy, Japan, Luxembourg, Netherlands, New Zealand, Norway, Portugal, Spain, Sweden, Turkey, United Kingdom, United States.
 \item[Cluster 3: ]  Afghanistan, Argentina, Bolivia, Brazil, Chile, Colombia, Costa Rica, Cuba, Dominican Republic, Ecuador, Egypt, El Salvador, Ethiopia, Guatemala, Haiti, Honduras, India, Iran, Lebanon, Mexico, Myanmar (Burma), Nicaragua, Pakistan, Panama, Paraguay, Peru, Philippines, Saudi Arabia, Syria, Thailand, Uruguay, Venezuela.
\item[Cluster 4: ] Bahamas, Bahrain, Bangladesh, Barbados, Benin, Bhutan, Botswana, Burkina Faso, Burundi, Chad, Congo - Brazzaville, Côte d’Ivoire, Eswatini, Fiji, Gabon, Ghana, Guinea, Guyana, Jamaica, Jordan, Kenya, Kuwait, Laos, Lesotho, Libya, Madagascar, Malawi, Malaysia, Maldives, Mali, Mauritania, Mauritius, Morocco, Nepal, Niger, Nigeria, Oman, Qatar, Senegal, Sierra Leone, Singapore, Sri Lanka, Sudan, Tanzania, Togo, Trinidad and Tobago, Tunisia, Uganda, United Arab Emirates, Zambia.
\item[Cluster 5: ] Armenia, Azerbaijan, Bosnia and Herzegovina, Eritrea, Kazakhstan, Kiribati, Kyrgyzstan, Marshall Islands, Micronesia (Federated States of), Montenegro, Namibia, Nauru, North Korea, Palau, South Sudan, Switzerland, Tajikistan, Timor-Leste, Tonga, Turkmenistan, Tuvalu, Uzbekistan, Yemen, Zanzibar.
\item[Cluster 6: ] Angola, Antigua and Barbuda, Belize, Brunei, Cape Verde, Comoros, Djibouti, Dominica, Grenada, Guinea-Bissau, Mozambique, Papua New Guinea, Samoa, São Tomé and Príncipe, Seychelles, Solomon Islands, St. Kitts and Nevis, St. Lucia, St. Vincent and Grenadines, Suriname, Vanuatu, Vietnam, Zimbabwe.
\item[Cluster 7: ] Andorra, Croatia, Czechia, Estonia, Georgia, Germany, Latvia, Liechtenstein, Lithuania, Moldova, Monaco, North Macedonia, San Marino, Slovakia, Slovenia, South Korea, Switzerland.
\end{itemize}

\end{document}